\RequirePackage{amsthm} 
\documentclass[pdflatex, sn-nature]{sn-jnl}   

\usepackage{graphicx}%
\usepackage{multirow}%
\usepackage{amsmath,amssymb,amsfonts}%
\usepackage{stmaryrd}
\usepackage{mathrsfs}%
\usepackage[title]{appendix}%
\usepackage{xcolor}%
\usepackage{textcomp}%
\usepackage{manyfoot}%
\usepackage{booktabs}%
\usepackage{algorithm}%
\usepackage{algorithmicx}%
\usepackage{algpseudocode}%
\usepackage{listings}%
\usepackage{comment}
\usepackage{caption}
\usepackage{mathtools}  
\usepackage[inline]{enumitem} 

\usepackage{mathbbol}
\usepackage{proof}

\renewcommand{\varphi}{\phi}

\newcommand{\seq}{\Rightarrow}
\newcommand{\base}[1]{\mathscr{#1}}
\newcommand{\baseB}{\base{B}}
\newcommand{\baseC}{\base{C}}
\newcommand{\baseD}{\base{D}}
\newcommand{\baseE}{\base{E}}

\newcommand{\baseX}{\base{X}}

\newcommand{\baseILL}{\base{N}}
\newcommand{\emptybase}{\varnothing}

\newcommand{\At}{\mathbb{A}}

\newcommand{\baseGeq}{\supseteq}
\newcommand{\baseLeq}{\subseteq}

\newcommand{\ILL}{ILL}
\newcommand{\provesILL}{\vdash}

\newcommand{\emptymultiset}{\varnothing}
\newcommand{\deriveBaseM}[1]{\vdash_{\!\!#1}}

\newcommand{\suppIPL}[1]{\Vdash_{ \!\!#1 }^{\!\!\mathfrak{S}}}
\newcommand{\suppM}[2]{\Vdash_{ \!\!#1 }^{ \!\!#2 }}

\newcommand{\suppL}[3]{\Vdash_{ \!\!#1 }^{{ \!\!#2 }\ctxt\hspace*{0.08cm}{ \!\!#3 }}}

\newcommand{\mand}{\otimes}
\newcommand{\mtop}{1}
\newcommand{\mto}{\multimap}
\newcommand{\aand}{\mathbin{\&}}
\newcommand{\aor}{\oplus}
\newcommand{\abot}{0}
\newcommand{\bang}{\mathop{!}}
\newcommand{\makeMultiset}[1]{\{#1\}}

\newcommand{\flatILL}[1]{{#1}^{\flat}}

\newcommand{\deflatILL}[1]{{#1}^{\natural}}
\newcommand{\openaddrule}{\{}
\newcommand{\closeaddrule}{\}}

\DeclareMathSymbol{\msetsum}{\mathrel}{bbold}{\lq\,}
\newcommand{\ctxt}{;}
\newcommand{\iplill}[1]{(#1)^\star}

\newcommand{\rn}[1]{\mathsf{#1}}
\newcommand{\ern}[1]{\rn{#1}_\mathsf{E}}
\newcommand{\irn}[1]{\rn{#1}_\mathsf{I}}

\newcommand{\promotion}{\rn{Prom}}
\newcommand{\dereliction}{\rn{Der}}
\newcommand{\weakening}{\rn{Wk}}
\newcommand{\contraction}{\rn{Ctr}}

\makeatletter
\def\labelandtag#1#2{\begingroup
   \def\@currentlabel{#2}%
   \phantomsection\label{#1}\endgroup
}
\makeatother

\def\descriptionlabel#1{\hspace\labelsep \upshape #1}
\makeatletter
\let\orgdescriptionlabel\descriptionlabel
\renewcommand*{\descriptionlabel}[1]{%
  \let\orglabel\label
  \let\label\@gobble
  \phantomsection
  \edef\@currentlabel{#1}%
  \let\label\orglabel
  \orgdescriptionlabel{(#1)}%
}
\makeatother

\theoremstyle{thmstylethree}
\newtheorem{theorem}{Theorem}[section]
\newtheorem{proposition}[theorem]{Proposition}%
\newtheorem{lemma}[theorem]{Lemma}%
\newtheorem*{lemma*}{Lemma}%
\newtheorem{corollary}[theorem]{Corollary}%
\newtheorem{definition}[theorem]{Definition}%
\newtheorem{defn}[theorem]{Definition}%

\begin{document}

\title[A Proof-theoretic Semantics for Intuitionistic Linear Logic]{\center{A Proof-theoretic Semantics for \\ Intuitionistic Linear Logic}}

\author*{\fnm{Yll} \sur{Buzoku}$^{[0009-0006-9478-6009]}$\footnote{The author has received funding from the European Union’s Horizon 2020 research and innovation programme under the Marie Skłodowska-Curie grant agreement No 101007627.}} \email{y.buzoku@ucl.ac.uk}
\affil{\orgdiv{Department of Computer Science}, \orgname{University College London}, \orgaddress{\street{66-72 Gower Street}, \city{London}, \postcode{WC1E 6EA}, \country{United Kingdom}}}

\abstract{The approach taken by Gheorghiu, Gu and Pym in their paper on giving a base-extension semantics for Intuitionistic Multiplicative Linear Logic is an interesting adaptation of the work of Sandqvist for IPL to the substructural setting. What is particularly interesting is how naturally the move to the substructural setting provided a semantics for the multiplicative fragment of intuitionistic linear logic. Whilst ultimately the Gheorghiu, Gu and Pym used their foundations to provide a semantics for bunched implication logic, it begs the question, what of the rest of intuitionistic linear logic? In this paper, I present just such a semantics. This is particularly of interest as this logic has as a connective the bang, a modal connective. Capturing the inferentialist content of formulae marked with this connective is particularly challenging and a discussion is dedicated to this at the end of the paper.}


\maketitle

\section{Introduction}\label{sec:Introduction}

Proof-theoretic Semantics (P-tS) is an alternative approach to semantics in which proof, rather than truth, takes a central role in conferring meaning to logical expressions. This can be seen as a mathematical realisation of the philosophical paradigm of Inferentialism; a position which seeks to determine the meanings of expressions through their use. In this realisation, the meanings of logical expressions are given through some notion of proof. This philosophical position has its origins in the works of Wittgenstein~\cite{Wittgenstein1953-WITPI-4} in which he says that the meaning of a word should be determined by its use, though more recent works, such as those of Brandom~\cite{brandom_ArticulatingReasons,brandom_MakingItExplicit,Brandom_ReasonsForlogic} and Dummett~\cite{Dummett_LogicalBasis_1991} have also further explored this position.

Modern proof-theoretic semantics can be seen as having two main branches, both stemming from Prawitz's original idea of a General Proof Theory\footnote{Which itself stems from Prawitz's considerations of Gentzen's statements on natural deduction in his paper~\cite{Gentzen_Untersuchungen_1935}.}~\cite{Prawitz1971ideas, prawitz1973towards}: The first, that of Proof-theoretic Validity (P-tV), concerns itself with the issue of what constitutes a valid proof. This approach to meaning can be seen as being closer to Prawitz's General Proof Theory than the alternative we present here and has been explored by authors such as Dummett~\cite{Dummett_LogicalBasis_1991}, Prawitz~\cite{Prawitz_General_Proof_Theory_1974}, Schroeder-Heister~\cite{Schroeder-Heister1991-SCHUPS,SchroederHeister_UniformPtSforLogicalConstants_1991} along with Piecha and de-Campos Sanz~\cite{Piecha2015failure, Piecha_CompletenessInPtS_2016}. The second, that of Base-extension Semantics (B-eS), concerns itself with the question of what constitutes a valid formula, and is the approach we will consider in this paper. This approach has been previously explored by Sandqvist~\cite{Tor2005, Tor_HypothesisDischarging_2015} and Schroeder-Heister, Piecha and de-Campos Sanz~\cite{Piecha_CompletenessInPtS_2016, CamposSanz2014-CAMACR-5}.
Whilst Sandqvist gives a sound and complete B-eS for classical logic in his doctoral thesis~\cite{Tor2005}, which he also later also discusses in~\cite{Tor2009}, it was really his work on Intuitionistic Propositional Logic (IPL)~\cite{Tor2015} which cemented the importance of B-eS as a viable approach to P-tS. In particular, in this work, he develops the rather elegant completeness argument that forms the basis of our own in this paper, which we shall discuss below. For a good comparison of these two approaches to Pt-S (as well as a third approach that is similar to B-eS), the reader is referred to~\cite{Antonio_Comparison2025}.

Sandqvist starts from the notion of an atomic rule $\mathcal{R}$, which we write linearly as $(P_1\seq q_1,\dots,P_n\seq q_n)\seq r$ and defines a base $\baseB$ to be a set of such rules. Such rules are to be considered as an instance of a valid inferential step that one can use in the justification of a particular atomic sentence. This is made precise through the definition of a derivability relation, $\deriveBaseM{\baseB}$, between bases, sets of atoms and individual atoms, as exemplified in Figure~\ref{fig:B-eS-IPL}. One should understand $\baseB$ as an $\rm NJ$-like object, whose elements are not schematic in nature. Consequently, just as how $\rm NJ$ can be thought of as generating the consequence relation $\vdash_{\rm NJ}$ for IPL, $\baseB$ can be thought of as generating the consequence relation $\deriveBaseM{\baseB}$. Sandqvist uses this consequence relation to thus define a support relation $\suppM{\baseB}{}$ on IPL sequents.

\begin{figure}[ht]
\hrule\small
\vspace{2mm}
\[
\frac{\begin{array}{ccc}
        [P_1] &        & [P_n] \\
        q_1   & \ldots & q_n  
      \end{array}
}{r} \, \mathcal{R}
\qquad 
{\begin{array}{rl}
\mbox{(Ref)} & \mbox{$P , p \vdash_\baseB p$} \\ 
(\mbox{App}_\mathcal{R}) & \mbox{if $((P_1 \Rightarrow q_1) , \ldots , (P_n \Rightarrow q_n)) \Rightarrow r)$ and,} \\
    & \mbox{for all $i \in [1,n]$, $P , P_i \vdash_\baseB q_i$, 
    then $P \vdash_\baseB r$} 
\end{array}}
\]
\[{
\begin{array}{rl@{\quad}rl} 
\mbox{(At)} & \mbox{for atomic $p$, $\Vdash_\baseB p$ iff $\vdash_\baseB p$} & 
    (\lor) & \mbox{$\Vdash_\baseB \phi \lor \psi$ iff, for every atomic $p$ and every } \\ 
    & & & \mbox{$\baseC \!\supseteq\! \baseB$, if $\phi \Vdash_\baseC p$ and $\psi \Vdash_\baseC p$, then $\Vdash_\baseC p$}\\
(\supset) & \mbox{$\Vdash_\baseB \phi \supset \psi$ iff $\phi \Vdash_\baseB \psi$} & (\bot) & 
    \mbox{$\Vdash_\baseB \bot$ iff, for all atomic $p$, $\Vdash_\baseB p$} \\
(\land) & \mbox{$\Vdash_\baseB \phi \land \psi$ iff $\Vdash_\baseB \phi$ and 
    $\Vdash_\baseB \psi$} & \mbox{(Inf)} & \mbox{for $\Gamma \neq \emptyset$, 
        $\Gamma \Vdash_\baseB \phi$ iff, for every $\baseC \supseteq \baseB$,}  \\ 
    & & & \mbox{if $\Vdash_\baseC \gamma$, for every $\gamma \in \Gamma$, then $\Vdash_\baseC \phi$}     
\end{array}}
\] \vspace{-2mm}
\caption{Sandqvist's B-eS for Intuitionistic Propositional Logic}
\label{fig:B-eS-IPL} 
\vspace{2mm}
\hrule
\end{figure}  

The need to consider base extensions in the semantics arises from the requirement that the meaning of a sentence should be given by some notion of a  construction of it. In B-eS, one generally equates the concept of contruction of a sentence to that of the sentence satisfying the support relation relative to a base, with the base witnessing the construction of the sentence and thus, containing the ``meaning'' of the sentence\footnote{That the support relation is a satisfactory notion of construction is an issue that we shall not discuss in great detail in this paper. However I feel it prudent to make at least the following point: If one considers the case of atomic sentences, then we see that for the support relation to be satisfied, we require that the atomic sentence be provable in the base. Since, in Sandqvist's semantics for IPL and in our work, as we shall see later, the support relation is a conservative extention of the atomic derivability relation, it feels fair to say that the constructability of complex sentences is indeed being inductively captured by the support relation.}. With this in mind, it seems reasonable that, in the case of intuitionistic implication, a construction of $\varphi\supset\psi$, when combined with a construction of $\varphi$, should yield a construction of $\psi$. Without the base extension, we would be required to accept that a construction of $\psi$ may be obtained \emph{even without} a construction of $\varphi$. This would follow vacuously as, were we to define $\varphi\supset\psi$ according to the clause
\[
\suppM{\baseB}{}\varphi\supset\psi\text{ iff }\suppM{\baseB}{}\varphi \Rightarrow \,\suppM{\baseB}{}\psi
\]
then we would have that $\suppM{\baseB}{}\varphi\supset\psi$ and $\not\suppM{\baseB}{}\varphi$ imply $\suppM{\baseB}{}\psi$; something that we find quite undesirable in our constructive world view. On the other hand, by considering base extensions, we lose Prawitz's original idea, that bases are supposed to fix the meaning of sentences, as the meaning of sentences is now allowed to change. These issues are well discussed by Prawitz in~\cite{Prawitz1971ideas} but again more recently by Sandqvist~\cite{Tor2005, Tor_HypothesisDischarging_2015,Tor_5thPtSSymposiumTalk}. It is important to note that, in principle, the requirement for the extension is no different to the requirement that appears in Kripke semantics, where implication requires a condition on all accessible worlds\footnote{And in fact, base extensions in the semantics play a similar role to hypothetical derivations in natural deduction.}. It is with a support relation that one can give a semantics to the full logic. We briefly show how by giving an overview of the soundness and completeness arguments used by Sandqvist in~\cite{Tor2015}.

As mentioned, the semantics in Figure~\ref{fig:B-eS-IPL} are sound and complete with respect to IPL. In this case, an IPL sequent $(\Gamma : \varphi)$ is defined to be valid if and only if $\Gamma\suppM{\baseB}{}\varphi$ for all bases $\baseB$. We see this notion of validity is very similar to the usual conception of validity for Kripke semantics, where a formula $\varphi$ is considered to be valid if and only if for all models $\mathcal{M}$, we have $\mathcal{M}\vDash \varphi$. This similarity is discussed in detail by Schroeder-Heister in~\cite{Schroeder2007modelvsproof}. Though we have, until now, emphasised the similarities between the usual Kripke semantics for IPL and the B-eS of Figure~\ref{fig:B-eS-IPL}, it is imperative to note, however, that models and bases \emph{are not} the same thing. A sceptic of this claim need only look at the form of the definitional clauses. Disjunction, for example, can be seen to be defined very closely to the interpretation one gives the disjunction elimination rule of $\rm NJ$, in contrast to the usual meta-logical disjunction seen in the Kripke semantics. This seemingly contradicts the famous quote by Gentzen~\cite{Gentzen_Investigations_1964} that ``the introductions present, so to speak, the `definitions' of the symbols concerned". For details on why this definition of disjunction is necessary in B-eS, the reader is referred to~\cite{Tor2015,Piecha2015failure,Alex_FromPtV2BeS_2022,Prawitz_NatDed_2006,Pym-Ritter-Robinson_CategoricalP-tS_2025}.

Let us now sketch the arguments for soundness and completeness.

\begin{theorem}[IPL Soundness]
    If $\Gamma\vdash_{\rm NJ}\varphi$ then $\Gamma\suppM{\baseB}{}\varphi$, for all $\baseB$.
\end{theorem}
The proof of this theorem amounts to using the fact that derivations in $\rm NJ$ are inductively defined and so it suffices to show that if the hypothesis of every rule of $\rm NJ$ is valid then the conclusion also is. For example, in the case of $\irn \land$, we suppose that $\Gamma\suppM{\baseB}{}\varphi$ and $\Gamma\suppM{\baseB}{}\psi$ for all bases $\baseB$ and show that therefore $\Gamma\suppM{\baseB}{}\varphi\land\psi$ for all bases $\baseB$, as required.

\begin{theorem}[IPL Completeness]
    If $\Gamma\suppM{\baseB}{}\varphi$, for all $\baseB$ then $\Gamma\vdash_{\rm NJ}\varphi$.
\end{theorem}

Since we start from the hypothesis that $\Gamma\suppM{\baseB}{}\varphi$ holds in all $\baseB$, we can therefore work in a tailor-made base, $\baseILL$, whose purpose will be to simulate the rules of $\rm NJ$ in a particular way. Since bases cannot contain rule schemas\footnote{By definition, but this so because a rule scheme over atoms would degenerate the way rules confer meaning onto atomic sentences.}, we must therefore carefully pick atoms to represent all of the subformulae of the elements of the set $\{\alpha\,|\,\alpha\in\Gamma\cup\varphi\}$. After doing so, we construct $\baseILL$ to have all the rules of $\rm NJ$ instantiated with our specially chosen atoms at all positions and show that each such atom is, in fact, inductively defined by atomised versions of the definitional clauses of the formulae being represented. It then follows that any $\deriveBaseM{\baseILL}$ proof of any one of our specially chosen atoms corresponds to an $\rm NJ$ proof of the formula it is representing, thus giving the completeness argument.

Intuitionistic linear logic (ILL)~\cite{Girard_LinearLogic_1987} is the intuitionistic fragment of Linear Logic~\cite{Girard_LinearLogic_1987,girard_LLSyntaxAndSemantics}, a logic introduced by J.-Y. Girard. Linear logic is characterised by its feature that the general usage of the structural rules is strictly forbidden; their use instead being assigned to two ``structural modalities'', both having structural rules on only one side of the sequent (and an S4-like Necessitation rule on the other). This restriction on the general use of structural rules leads to a splitting of the connectives of classical logic into additive and multiplicative parts. For example, the classical connective $\land$ becomes split into the linear connectives $\mand$ and $\aand$; the former governing multiplicative reasoning and the later additive. However, more importantly, since we cannot generally weaken or contract, the notion of reflexivity of the consequence relation of the logic is now sufficiently strict as to require that $\Gamma\vdash\varphi$ if and only if $\Gamma=\{\varphi\}$. Were $\Gamma$ to contain any other formulae, this would amount to some measure of weakening. As a result, each formula can only be used once in a proof, unless there are multiple instances of it as hypothesis. This means that formulae in a proof act rather like resources that are being produced/consumed and that proofs keep track of which resources are needed to obtain what. This informal interpretation of proofs in linear logic is what has been dubbed the ``resource interpretation" of linear logic and will play a key role in guiding our choices in the sequel. We, however, will be concerned with the intuitionistic fragment of linear logic, whose consequence relation can be understood, as normal, by taking the usual sequent calculus of linear logic~\cite{Girard_LinearLogic_1987} and restricting the succeedent to contain at most one formula. This restriction trivialises the additive disjunction, the modality which has its structural rules on the right and the unit of the additive disjunction. The resulting system, whilst losing the nice symmetry properties previously exhibited across all connectives, still retains the resource interpretation. It also now becomes easier to introduce a natural deduction system for linear logic as in~\cite{Benton_Bierman_dePaiva_TermCalc4ILL_1993,Bierman_OnILL_1994,Mints_NormalNatDedForILL_1998,Troelstra_NatDetForILL_1992}. 

The purpose of this paper will be to give a base-extension semantics for intuitionistic linear logic in a manner similar to that introduced by Sandqvist for IPL. Doing so, however, is not a trivial task. As noted, the consequence relation of ILL is much more ``sensitive'' to structurality. For example, whereas $p$ is a consequence of the multiset $\{p,q\}$ in IPL, the same does not hold true in ILL. Thus, it follows that any proposed support relation for ILL cannot validate intuitionistically valid sequents such as $(\{p,q\}:p)$. Therefore, the matter at hand is more complex than simply studying a larger set of connectives using the setup of Sandqvist for IPL. Indeed, we need to somehow redefine the support relation to prevent it from validating sequents of this form. To this end, we appeal to the semantics of Gheorghiu, Gu, and Pym~\cite{AlexTaoDavid_PtS4IMLL} on the base-extension semantics for the Intuitionistic Multiplicative fragment of Linear Logic (IMLL). In this paper, the authors successfully redefine the support relation to correctly handle the issues around structurality discussed above. However, this isn't the only modification they make. We summarise their semantics in Figure~\ref{fig:B-eS-IMLL} and note another key difference; their interpretation of atomic rules is subtly, but crucially, different to that of Sandqvist's for IPL. In the semantics for IMLL, we see that atomic rules, when interpreted by atomic derivability (that is, $\deriveBaseM{\baseB}$) for IMLL, requires that each ``branch''  of an atomic rule ``brings'' its own multiset of open atomic assumptions, with the conclusion of the rule following as a consequence of the multiset union of all the multisets of open assumptions. That is to say, the $(\text{App}_{\mathcal{R}})$ clause differs considerably and meaningfully between the semantics for IMLL and IPL.

\begin{figure}[ht]
\hrule\small
\vspace{2mm}
\[
\frac{\begin{array}{ccc}
        [P_1] &        & [P_n] \\
        q_1   & \ldots & q_n  
      \end{array}
}{r} \, \mathcal{R}
\qquad 
{\begin{array}{rl}
\mbox{(Ref)} & \mbox{$p \vdash_\baseB p$} \\ 
(\mbox{App}_\mathcal{R}) & \mbox{if $((P_1 \Rightarrow q_1) , \ldots , (P_n \Rightarrow q_n)) \Rightarrow r)$ and,} \\
    & \mbox{for all $i \in [1,n]$, $C_i , P_i \vdash_\baseB q_i$, 
    then $C_1,\dots,C_n \vdash_\baseB r$} 
\end{array}}
\]
\[{
\begin{array}{rl@{\quad}rl} 
\mbox{(At)} & \mbox{for atomic $p$, $\Vdash^L_\baseB p$ iff $L\vdash_\baseB p$} & 
    (\otimes) & \mbox{$\Vdash^L_\baseB \phi \otimes \psi$ iff, for all atomic $p$, all $\baseC \!\supseteq\! \baseB$ and} \\ 
    & & & \mbox{multiset of atoms $K$, if $\phi,\psi \Vdash^K_\baseC p$, then $\Vdash^{L,K}_\baseC p$}\\[2mm]
(\multimap) & \mbox{$\Vdash^L_\baseB \phi \multimap \psi$ iff $\phi \Vdash^L_\baseB \psi$} & (I) & 
    \mbox{$\Vdash^L_\baseB I$ iff, for all $\baseC \!\supseteq\! \baseB$, atomic $p$ and} \\
    & & & \mbox{multiset of atoms $K$, if $\Vdash_{\baseC}^Kp$ then $\Vdash_{\baseC}^{L,K}p$}\\[2mm]
    & & \mbox{(Inf)} & \mbox{for $\Gamma \neq \emptyset$, 
        $\Gamma \Vdash_\baseB \phi$ iff, for every $\baseC \supseteq \baseB$,}  \\ 
    & & & \mbox{if $\Vdash_\baseC \gamma$, for every $\gamma \in \Gamma$, then $\Vdash_\baseC \phi$}     
\end{array}}
\] 
\caption{Gheorghiu, Gu, and Pym's B-eS for IMLL}
\label{fig:B-eS-IMLL} 
\vspace{2mm}
\hrule
\end{figure} 

If one takes a look at the natural deduction system for IMLL, one sees that indeed, this is expected, for this is the interpretation we give to the rules. In fact, if we formally spell out this interpretation of derivability over the natural deduction system for IMLL, we obtain individual instances of something very similar to the ($\text{App}_{\mathcal{R}}$) clause in Figure~\ref{fig:B-eS-IMLL}, with the only differences being that we consider formulae and multisets thereof instead of just atoms. However, IMLL, being a purely multiplicative logic, has the obvious side effect that it doesn't account for additive or modal behaviours. This is the situation we find ourselves in, for ILL has rules which are multiplicative, additive, contains elements of both, and has rules for a modality. Thus, it is to be expected that any proof-theoretic semantics for ILL, in the style we are interested in, must also modify its notion of atomic rule and atomic derivability, as compared to that of IMLL, lest it forgo completeness. 

Thus, we are left with trying to develop a notion of atomic rule and application of such rules, that mimics that which we have in some natural deduction system for ILL. It is at this point that we are faced with a problem. What should the general form of an inference figure in a natural deduction system for ILL take? This is a serious problem which is, to the knowledge of the author, never adequately asked nor answered in the literature. In this paper, we additionally attempt to address this issue, though the solution is, unfortunately, somewhat convoluted.

Rules in natural deduction systems for ILL tend to be 
defined in a fairly ad-hoc manner. This is a completely acceptable approach to presenting such systems as, in reality, there are only a few number of rules in such systems and so, if one can provide a nice series of rule schemas, a few of which have side conditions to ensure that the rules are interpreted properly, then one can use these systems to study the proof-theoretic properties of such logics, taking care to ensure that the few side conditions are adhered to. These side conditions can be explicitly seen in systems such as those found in~\cite{Bierman_OnILL_1994,Negri_NormalisingNatDedForILL_2002,Mints_NormalNatDedForILL_1998}, to list a few. The main side condition we are talking of here is, of course, that of the interpretation of the contexts in the rules, particularly in the case of the rules governing the behaviour of the additive connectives. The side conditions tend to be justified by additionally giving a sequent-style presentation of the natural deduction figures, where multisets of open assumptions are baked into the notation using explicit schematic context metavariables. Indeed, there is nothing wrong with doing so and, in fact, we shall do something similar in Section~\ref{sec:ILL}. Nevertheless, we show that it is indeed possible to give a notion of derivability which handles additivity systematically, without side-conditions or appeal to explicit context metavariables, using a device which we call an \emph{additive box}; a device which really acts as a scope delimiter for contexts (or more formally, multisets of open assumptions)\footnote{In fact, by using additive boxes in the proof-theory, one can formalise a reinterpretation of the context metavariables as labels on branches instead of sets of arbitrary open assumptions. Under this reinterpretation, all branches labelled with the same context metavariable are considered additively. However, using additive boxes directly is more expressive than labelling branches as it allows for ``empty'' additive boxes, something the ``branch-first'' approach usually taken with working with inference rules in natural deduction cannot do, as you cannot give a branch label to a non-existent branch. This point will be made clearer in Section~\ref{sec:ILL}, as it turns out, the rules $\ern\abot$ and $\irn\top$ require such empty additive boxes.}.

From the point of view of developing a B-eS for ILL, since we would like to develop a notion of atomic derivability that is \emph{as similar as possible} to derivability in some natural deduction system for ILL, being able to eschew ourselves from using rules with context metavariables in the proof-theory is very desirable, as rules of inference in B-eS generally do not speak of sets of open assumptions. Doing so would effectively mean that certain inferences are ``context specific'', resigning us to the position that such inferences can only be made given that we have a particular multiset of open assumptions. Given that ILL has no ``context-specific'' rules of this kind, as the context metavariables are always arbitrary multisets of formulae, it would be rather odd were our B-eS for ILL to require such constructions. Whilst it is plausible that for some logics this is indeed necessary, as we show in this paper, Intuitionistic Linear Logic is not one such logic.

However, expressing additivity systematically is not the only problem we have with natural deduction systems for ILL. Indeed, such systems tend to have rules with variably many minor premisses; something not usually seen in natural deduction. Such rules seem to be degenerating infinitely many rules into a single inference figure. This leads us to ask the question: Is an instance of such a rule a substitution instance of formulae \emph{and} a number of premisses or just a substitution instance for formulae and any number of premisses? As we shall see in Section~\ref{sec:ILL} (and also, as discussed by Negri in~\cite{Negri_NormalisingNatDedForILL_2002}), the variable number of minor premisses is really just a heavy handed way of saying that we require any arbitrary multiset of assumptions to hold at the hypothesis line. It turns out that additive boxes once again come to the rescue here, using the notion of an empty additive box, in all cases except for that of the promotion rule (c.f. Figure~\ref{fig:NatDedILL}), where one has to define a special type of additive box that we call a modal box. The modal box internalises the additional side conditions that are required when treating the modality of ILL in natural deduction, at the cost of having to be treated as a separate type of assumption. Nevertheless, by abstracting these side conditions appropriately, using these ``boxes'', one obtains a clean and systematic representation for the inference figures of the natural deduction system for ILL, where the answer to the question of what the general form of an inference figure in a natural deduction system for ILL is easily answerable. Thus, in Section~\ref{sec:ILL}, we spend considerable time developing this new language. One unfortunate side effect of this effort, however, is that we will require a somewhat complexly defined derivability relation on ILL sequents, to tell us how to interpret the rules. Nevertheless, as alluded to previously, a positive result of this effort will be that we develop a definition atomic derivability that is very similar to that for ILL; similar enough in fact, that the proof of completeness of the semantics should not be too complicated, as initially desired.

Before moving on, it is worth noting that whilst Gheorghiu, Gu, and Pym have extended their result to give a B-eS for the logic of Bunched Implications~\cite{AlexTaoDavid_PtS4BI}, at present, there is very little work done on understanding modalities from a base-extension semantics (or even more generally from the point of view of proof-theoretic semantics), and even less to understanding additivity. Some work in this direction has been undertaken by Gheorghiu, Gu, and Pym, with the goal to understand the role of resource from the perspective of base-extension semantics~\cite{AlexTaoDavid_ResourceSemantics_2024, AlexTaoDavid_ResourceSemanticsNote_2024}. K\"{u}rbis, for example, notes in his paper~\cite{Kurbis_PtSNegAndModality_2015} some conditions in his view toward a theory of modalities from the point of view of proof-theoretic semantics. However the framework he considers is considerably different from that which we consider here. Nontheless, it is Eckhardt and Pym~\cite{TimoDavid_PtSModalS4,TimoDavid_PtSModalS5} who first developed a base-extension semantics for the classical modal logics K, KT, K4, S4 and S5 in a manner closer to Sandqvist's original work for classical logic, with Buzoku and Pym~\cite{BuzokuPymIMLs2025} having developed semantics for the intuitionistic modal logics defined by Simpson~\cite{simpson1994Thesis} in a style much closer to that which we are interested in here. Unfortunately, both of these approaches, in the opinion of the author, leave something to be desired. Whether it is due to the fact that modalities generally have poor proof-theoretic properties or simply that we are lacking some deep metaphysical insight into what it means for a word of a language to be modal, with this paper, it is hoped that we will be able to at least address the issues I feel are present in the approaches previously taken towards understanding modalities from the perspective of base-extension semantics. This will be done by taking a much more intrinsically proof-theoretic approach to understanding modality in the setting of Intuitionistic Linear Logic. Our paper thus goes as follows: We start by giving an overview and introducing a new syntax for natural deduction for Intuitionistic Linear Logic, in Section~\ref{sec:ILL}. We then introduce the key semantic structures, that of atomic derivability, in Section~\ref{sec:Basic-Derivability}, and the support relation, in Section~\ref{sec:BeS}, and proceed to prove the necessary structural properties of both. We then have the key results of this paper, that of the soundness and completeness results, in Sections~\ref{sec:Soundness} and~\ref{sec:Completeness} respectively, before finally finishing with an overview of the results contained herein in Section~\ref{sec:Conc}.

\section{Intuitionistic Linear Logic}\label{sec:ILL}
For the remainder of the paper we will assume a fixed set of propositional atoms $\At$ that we refer to interchangeably as atoms or basic sentences. Unless stated otherwise, lowercase latin letters will be used to refer to atoms with uppercase latin letters being used to refer to finite multisets thereof. Similarly, lowercase greek letters will be used to represent individual formulae of Intuitionistic Linear Logic with uppercase letters being used to represent finite multisets thereof. We suppress set theoretic notation in the usual way, with the caveat that we write the multiset union of two multisets $\Gamma$ and $\Delta$ as $\Gamma\msetsum\Delta$. That is to say, if $\Gamma = \{a,a,b\}$ and $\Delta = \{b,c,a\}$ then $\Gamma\msetsum\Delta = \{a,a,a,b,b,c\}$. Finally, throughout this paper, the term ``atomic multiset" is taken to mean multiset of propositional atoms.

\begin{definition}[Intuitionistic linear formulae]\label{def:ILL-formula}
    Formulae of \ILL{} are defined by the grammar: $\phi ::= p \in \At \mid \top \mid \abot \mid \mtop \mid \phi \mto \phi \mid \phi \mand \phi \mid \phi \aand \phi \mid \phi \aor \phi \mid \bang \phi$
\end{definition}

We additionally may write $\bang\Gamma$ to mean a multiset of ILL formulae where the top-level connective is $(\bang)$; that is to say, for all $\alpha\in\bang\Gamma$, there exists an ILL formula $\gamma$, such that $\alpha=\bang\gamma$.


\begin{definition}[Sequent]\label{def:ILL-sequent}
    An intuitionistic linear sequent (or just sequent) is an ordered pair $\langle \Gamma, \phi \rangle$ which we write as $(\Gamma : \phi)$, where $\Gamma$ is a (finite) multiset of \ILL{} formulae and $\phi$ is a single \ILL{} formula. For visual clarity, we may write $(\Gamma:\varphi)$ as $\Gamma\seq\varphi$.
\end{definition}

\begin{definition}[Intuitionistic linear derivability]\label{def:ILL-derivability}
Given the sequent $\Gamma\seq\varphi$, the relation of derivability, $\provesILL$, is defined inductively according to the schemas of Figure~\ref{fig:NatDedILL}. The resulting consequence relation is written as $\Gamma\provesILL\phi$.
\end{definition}

\begin{figure}[t]\captionsetup{justification=centering,font=small}
    \hrule \vspace{2mm}
    \small\[
    \begin{array}{c}
         \infer[Ax]{\varphi\provesILL\varphi}{}
    \end{array}
    \]
    \small\[{
    \begin{array}{c@{\quad}c} 
    \infer[\irn \mto]{\Gamma\provesILL\phi\mto\psi}{\Gamma\msetsum\phi\provesILL\psi} & \infer[\ern \mto]{\Gamma\msetsum\Delta\provesILL\psi}{\Gamma\provesILL\phi\mto\psi \quad \Delta\provesILL\phi} \\[3mm]
    \infer[\irn \mand]{\Gamma\msetsum\Delta\provesILL\phi\mand\psi}{\Gamma\provesILL\phi \quad \Delta\provesILL \psi} & \infer[\ern \mand]{\Gamma\msetsum\Delta\provesILL\chi}{\Gamma\provesILL\psi\mand\psi \quad \Delta\msetsum\phi\msetsum\psi\provesILL\chi}\\[3mm]
    \infer[\irn \mtop]{\provesILL\mtop}{} & \infer[\ern \mtop]{\Gamma\msetsum\Delta\provesILL\phi}{\Gamma\provesILL\mtop \quad \Delta\provesILL\phi}\\[3mm]
    \infer[\irn \aand]{\Gamma\provesILL\phi\aand\psi}{\Gamma\provesILL\phi \quad \Gamma\provesILL\psi} & \infer[\ern{\aand_1}]{\Gamma\provesILL\phi}{\Gamma\provesILL\phi\aand\psi} \quad \infer[\ern{\aand_2}]{\Gamma\provesILL\psi}{\Gamma\provesILL\phi\aand\psi}\\[3mm]
    \infer[\irn{\aor_1}]{\Gamma\provesILL\phi\aor\psi}{\Gamma\provesILL\phi} \quad \infer[\irn{\aor_2}]{\Gamma\provesILL\phi\aor\psi}{\Gamma\provesILL\psi}  & \infer[\ern \aor]{\Gamma\msetsum\Delta\provesILL\chi}{\Gamma\provesILL\phi\aor\psi \quad \Delta\msetsum\phi\provesILL\chi \quad \Delta\msetsum\psi\provesILL\chi}\\[3mm]
    \infer[\irn \top]{\Gamma_1\msetsum\dots\msetsum\Gamma_n\provesILL\top}{\Gamma_1\provesILL\phi_1\,\dots\,\Gamma_n\provesILL\phi_n} & \infer[\ern \abot]{\Gamma_1\msetsum\dots\msetsum\Gamma_n\msetsum\Delta\provesILL\chi}{\Gamma_1\provesILL\phi_1\,\dots\,\Gamma_n\provesILL\phi_n \quad \Delta\provesILL\abot}\\[3mm]
    \infer[\promotion]{\Gamma_1\msetsum\dots\msetsum\Gamma_n\provesILL\bang\phi}{\Gamma_1\provesILL\bang\psi_1\,\dots\,\Gamma_n\provesILL\bang\psi_n \quad \bang\psi_1\msetsum\dots\msetsum\bang\psi_n\provesILL\phi} & \infer[\dereliction]{\Gamma\msetsum\Delta\provesILL\psi}{\Gamma\provesILL\bang\phi \quad \Delta\msetsum\phi\provesILL\psi}\\[3mm]
    \infer[\weakening]{\Gamma\msetsum\Delta\provesILL\psi}{\Gamma\provesILL\bang\phi \quad \Delta\provesILL\psi} & \infer[\contraction]{\Gamma\msetsum\Delta\provesILL\psi}{\Gamma\provesILL\bang\phi \quad \Delta\msetsum\bang\phi\msetsum\bang\phi\provesILL\psi}\\[3mm]
    \end{array}}
    \] 
    \vspace{-1mm}
    
    \caption{The natural deduction system $\rm N_{ILL}$ for Intuitionistic Linear Logic in sequent style. The rules $\promotion$, $\irn \top$ and $\ern \abot$ hold for all $n\geq 0$.} \vspace{2mm}
    \hrule
    \label{fig:NatDedILL}
\end{figure}

As described, Figure~\ref{fig:NatDedILL} presents the natural deduction system $\rm N_{ILL}$ in sequent style. A natural question to ask is whether it is possible to do so in a more traditional, Gentzen-Prawitz tree style? This issue is well covered in~\cite{Bierman_OnILL_1994}, but, as mentioned in the introduction, we wish to go further. If we consider the multiplicative fragment of \ILL{}, then it is clear that we can always just consider all branches of an inference figure to be multiplicative with respect to each other (that is, that they have disjoint contexts) and so we naturally obtain such a calculus. Thus, an inference figure such as 

\[
\begin{array}{ccc}
\infer[\irn \mand]{\Gamma\msetsum\Delta\provesILL\phi\mand\psi}{\Gamma\provesILL\phi & \Delta\provesILL \psi} & \text{ becomes } & \infer[\irn \mand]{\phi\mand\psi}{\phi & \psi}\\[3mm]  
\end{array}
\]

Such a system is used by the authors of~\cite{AlexTaoDavid_PtS4IMLL} to give their base-extension semantics for the multiplicative fragment of \ILL{}. If we include $(\bang)$ to this fragment, to get the mutliplicative-exponential fragment, then we need to introduce the notion of strict derivations to correctly encode the Promotion rule. This is because in the rule $\promotion$, there is the requirement that $\bang\psi_1\msetsum\dots\bang\psi_n\provesILL\varphi$ occurs \emph{without} any multiset of open assumptions, for all $n\geq 0$. We show this requirement using semantic brackets to indicate that applying the rule requires a proof of $\varphi$ from the discharge set $\bang\psi_1\msetsum\dots\bang\psi_n$. 
Thus, the tree-like inference figure for promotion becomes 
\[
\infer[\promotion]{\bang\phi}{\bang\psi_1\,\dots\,\bang\psi_n & \deduce{\phi}{\llbracket \bang\psi_1\msetsum\dots\msetsum\bang\psi_n\rrbracket}}
\]

A key point of note with regards to this rule however, is the presence of the $n$ subscript. This $n$ specifically requires that we are talking about \emph{any} and \emph{all} multisets of formulae prefixed with a ($\bang$) concurrently. An alternative, and more honest, way of writing this rule would be 

\[
\infer[\promotion]{\bang\phi}{\bang\Gamma & \deduce{\phi}{\llbracket \bang\Gamma\rrbracket}}
\]

Thus becomes apparent the real meaning of this rule; that $\varphi$ must follow from a context of formulae with $(\bang)$ as a top-level connective. This is nothing new, having been well investigated in~\cite{Benton_Bierman_dePaiva_TermCalc4ILL_1993,Bierman_OnILL_1994}, but in short, the schematic setting that we are usually in when considering inferences in Natural Deduction systems means that this sort of rule is not problematic, since we can range $\Gamma$ over all possible multisets and prefix each element of the context with a $(\bang)$. Furthermore, we intuitively understand what it means to assert $\bang\Gamma$ at the hypothesis line, but formally speaking, this would be inappropriate, thus leading to our initial characterisation of this rule. However, as discussed in the introduction, the inference figures we will be concerned with in the semantics are \emph{not} schematic in nature and that can therefore only contain \emph{known} propositional atoms; atoms which cannot have a $(\bang)$ as a top level connective, as doing so would mean the are no longer atoms! Furthermore, they cannot have a \emph{variable} number of premisses as it would no longer be a rule instance if they did. But individual instantiations of the promotion rule for any finite $n$ are insufficient, since that would mean that you only consider contexts of certain sizes; something that clearly isn't the case in the $\promotion$ rule. These restrictions in the semantics mean that, when defining atomic derivability, our usual way of expressing the promotion rule is insufficient. We therefore need a different characterisation of the promotion rule. To this end, we redefine the meaning of the semantic bracket, and write the promotion rule now as follows:

\[
\infer[\promotion]{\bang\phi}{\llbracket \varphi\rrbracket}
\]

We call the semantic bracket $\llbracket\cdot\rrbracket$ here a modal box. This rule is to be operationally interpreted as saying:
\begin{enumerate}
    \item If there is a derivation of $\varphi$ from some, possibly empty, multiset of formulae $\{\alpha_1,\dots,\alpha_n\}$ i.e. $\alpha_1\msetsum\dots\msetsum\alpha_n\provesILL\varphi$
    \item Each $\alpha_i$ is a formula which is the conclusion of some instance of a rule with only modal boxes above the inference line
    \item We have a derivation for each $\alpha_i$, i.e. $\Gamma_i\provesILL\alpha_i$.
\end{enumerate}
\vspace{2mm}
then it follows that $\Gamma_1\msetsum\dots\msetsum\Gamma_n\provesILL\bang\varphi$. That is to say, the behavioural reading of the promotion rule is exactly as it was before, including the strict derivation (for a short discussion on why this is necessary, c.f.~\cite{Bierman_OnILL_1994}), since the only rule which satisfies condition $(2)$ is $\promotion$ and thus the only formulae which $\alpha_i$ can be are formulae with $(\bang)$ as a top-level connective. However, note that at no point did we define the multiset of formulae $\{\alpha_1,\dots,\alpha_n\}$ as being a multiset of formulae with $(\bang)$ as a top-level connective. This point is crucial, as we now have a characterisation of the $\promotion$ rule that is not defined in terms of any connectives but in terms of some structural property of the rule itself. Note, that the operational reading of the rule is exactly the same as before, the only thing that has changed is the way we express this operation, as we shall see in Theorem~\ref{thm:ILL-derivability-equivalence} below.

So what of the additives? In this case, the author of~\cite{Bierman_OnILL_1994} shows that to consider the additives, the language we have developed so far gets us close. For example, we may represent the additive conjunction as 
\[
\infer[\irn \aand]{\varphi\aand\psi}{\chi_1\dots\chi_n & \deduce{\varphi}{\llbracket\chi_1\msetsum\dots\msetsum\chi_n\rrbracket} & \deduce{\psi}{\llbracket\chi_1\msetsum\dots\msetsum\chi_n\rrbracket}}
\]
where the semantic brackets here mean that we discharge all assumptions at once (as in~\cite{Bierman_OnILL_1994}) and that there are no other open assumptions on that branch, as in the original formulation of the $\promotion$ rule. Our interpretation of this rule is that that we discharge both contexts $\chi_1\msetsum\dots\msetsum\chi_n$ and re-introduce it but once. Whilst technically such a presentation is fine, it is at odds with our natural conception of additivity. 
Similarly, we can write this rule more honestly as 
\[
\infer[\irn \aand]{\varphi\aand\psi}{\Gamma & \deduce{\varphi}{\llbracket\Gamma\rrbracket} & \deduce{\psi}{\llbracket\Gamma\rrbracket}}
\]
We won't go into details again but as in the case of the Promotion rule, we take problem with rules of this form as they require that we consider arbitrary contexts at the hypothesis line\footnote{Though this point will not be explored further in this paper, the fact that additive rules can be expressed so similarly to the promotion rule is, in the opinion of the author, a strong basis for the argument that additivity is somewhat of a ``modal'' concept.}.
Indeed, the author of~\cite{Bierman_OnILL_1994} instead opts to extend the proof-theory to include the concept of additive contexts. With additive contexts, the author is then able to represent inference rules which require context sharing derivations. They do so, effectively, by labelling each branch with a unique meta-variable which explicitly represents the context multiset used by those derivations. In this scheme, the $\irn\aand$ rule becomes
\[
\infer[\irn\aand]{\varphi\aand\psi}{\deduce{\varphi}{\Gamma} & \deduce{\psi}{\Gamma}}
\]

It is implicitly understood that the $\Gamma$ in both branches is the same $\Gamma$ \footnote{Some authors label these $\Gamma$'s to make this point explicit.} and that this is a context multiset. This notation, whilst clearly functional, is somewhat undesirable for it requires us to assign meta-variables to represent contexts when discussing inferences. Therein lies the problem; the rule requires that we talk about shared contexts in terms of derivations from arbitrary contexts $\Gamma$. The problem here is the requirement of the $\Gamma$ to represent an arbitrary context\footnote{This is different to the issue of discharging, for there we are justified in having discharge be a part of the rule for we specify precisely what we discharge. Here, we have an arbitrary (multi)set of assumptions required to correctly represent the structure of the rule.}. In a derivation, such $\Gamma$'s naturally accumulate and may be consumed through discharge. However, the rule uses these contexts effectively as labels. By doing so, we end up drawing vertical sequents and blur the lines between a pure rule of inference and its application. Since our rules are schematic in nature, this in itself isn't problematic. However, as we shall see in Section~\ref{sec:Basic-Derivability} and beyond, if one were to consider a non-schematic system, this distinction becomes important. Of course, we could simply ignore this issue and use the initial method of encoding additivity. However, the problems initially identified remain. Thus, for the remainder of the paper, we sahll use the system shown in Figure~\ref{fig:NatDedILL2} to better represent the rule schemas of $\rm N_{ILL}$. That isn't to say that this presentation of the rules constitutes a new calculus; this is simply a new way of systematically and uniformally representing the rules of the natural deduction system $\rm N_{ILL}$.

In this presentation, no additional meta-variables are used to represent additive contexts. Instead, derivations are marked as being within ``shared" contexts (demarcated by curly brackets), which we call additive boxes. This is what corresponds to an additive context in~\cite{Bierman_OnILL_1994}. All derivations within an additive box must share the same multiset of open assumptions (i.e. must be ``additive" with respect to each other) and, in the context of a rule, only one copy of the open assumptions from each additive box is understood to be necessary to obtain a derivation of the conclusion of a rule. If a rule has multiple additive boxes, then each additive box is understood to have disjoint contexts from every other additive box (i.e. the additive boxes are multiplicative with respect to each other). Thus, this syntaxt extends the tree like representation used for the purely multiplicative fragment of \ILL{}, without the need for complex interpretation as in the case of the modal box. For example, the rule 
\[
\begin{array}{ccc}
    \infer[\irn \mand]{\varphi\mand\psi}{\varphi&\psi} 
    & \text{ becomes } &
    \infer[\irn \mand]{\varphi\mand\psi}{\openaddrule\varphi\closeaddrule & \openaddrule\psi\closeaddrule}
\end{array}
\]

To see how these brackets represent additivity, let us now consider some rules governing additive connectives. The case of $\irn \aand$ for example gives that 
\[
\begin{array}{ccc}
    \infer[\irn \aand]{\varphi\aand\psi}{\deduce{\varphi}{\Gamma} & \deduce{\psi}{\Gamma} } 
    & \text{ becomes } &
    \infer[\irn \aand]{\varphi\aand\psi}{\openaddrule\varphi & \psi\closeaddrule}
\end{array}
\]

This notation is very flexible as it allows us to represent even $\ern \aor$ very naturally as
\[
    \infer[\ern \aor]{\chi}{\openaddrule\phi\aor\psi\closeaddrule & \left\{\raisebox{-0.5em}{ \deduce{\chi}{[\phi]} \quad \deduce{\chi}{[\psi]}} \right\}}
\]
We now give a formal definition of our characterisation of natural deduction for ILL using this new notation. Note that henceforth, the semantic brackets will be used to refer to modal boxes \emph{only}. An important point to note is that, since this notation is nothing more than formalisation of the presentation of the system presented in Figure~\ref{fig:NatDedILL}, the system we now formally define below continues to enjoy the exact same meta-logical properties of the system presented in Figure~\ref{fig:NatDedILL}, that is, the system observes the subformula property and is strongly normalising (as shown in~\cite{Bierman_OnILL_1994}); a fact that is a consequence of Theorem~\ref{thm:ILL-derivability-equivalence}. 

\begin{figure}
[t]\captionsetup{justification=centering,font=small}
    \hrule \vspace{2mm}
    \small\[{
    \begin{array}{c@{\quad}c} 
    \infer[\irn \mto]{\phi\mto\psi}{\left\openaddrule\raisebox{-0.5em}{\deduce{\psi}{[\phi]}}\right\closeaddrule} & \infer[\ern \mto]{\psi}{\openaddrule\phi\mto\psi\closeaddrule & \openaddrule\phi\closeaddrule} \\[3mm]
    \infer[\irn \mand]{\phi\mand\psi}{\openaddrule\phi\closeaddrule & \openaddrule\psi\closeaddrule} & \infer[\ern \mand]{\chi}{\openaddrule\phi\mand\psi\closeaddrule & \left\openaddrule\raisebox{-0.5em}{\deduce{\chi}{[\phi\msetsum\psi]}}\right\closeaddrule}\\[3mm]
    \infer[\irn \mtop]{\mtop}{} & \infer[\ern \mtop]{\phi}{\openaddrule\mtop\closeaddrule & \openaddrule\phi\closeaddrule}\\[3mm]
    \infer[\irn \aand]{\phi\aand\psi}{\openaddrule \phi & \psi \closeaddrule} & \infer[\ern{\aand_1}]{\phi}{\openaddrule\phi\aand\psi\closeaddrule} \quad \infer[\ern{\aand_2}]{\psi}{\openaddrule\phi\aand\psi\closeaddrule}\\[3mm]
    \infer[\irn{\aor_1}]{\phi\aor\psi}{\openaddrule\phi\closeaddrule} \quad \infer[\irn{\aor_2}]{\phi\aor\psi}{\openaddrule\psi\closeaddrule}  & \infer[\ern \aor]{\chi}{\openaddrule\phi\aor\psi\closeaddrule & \left\{\raisebox{-0.5em}{ \deduce{\chi}{[\phi]} \quad \deduce{\chi}{[\psi]}} \right\}}\\[3mm]
    \infer[\irn \top]{\top}{\openaddrule\emptymultiset\closeaddrule} & \infer[\ern \abot]{\chi}{\openaddrule\emptymultiset\closeaddrule & \openaddrule\abot\closeaddrule}\\[3mm]
    \infer[\promotion]{\bang\phi}{\llbracket\varphi\rrbracket} & \infer[\dereliction]{\psi}{\openaddrule\bang\phi\closeaddrule & \left\openaddrule\raisebox{-0.5em}{\deduce{\psi}{[\phi]}}\right\closeaddrule}\\[3mm]
    \infer[\weakening]{\psi}{\openaddrule\bang\phi\closeaddrule & \openaddrule\psi\closeaddrule} & \infer[\contraction]{\psi}{\openaddrule\bang\phi\closeaddrule & \left\openaddrule\raisebox{-0.5em}{\deduce{\psi}{[\bang\phi\msetsum\bang\phi]}}\right\closeaddrule}\\[3mm]
    \end{array}}
    \] 
    \vspace{-1mm}
    
    \caption{An alternative representation of the natural deduction system $\rm N_{ILL}$ for Intuitionistic Linear Logic in tree style.}\vspace{2mm}
    \hrule
    \label{fig:NatDedILL2}
\end{figure}

\begin{definition}[Additive box]
    An additive box is a (possibly empty) multiset of sequents.
\end{definition}

\begin{definition}[Rule schema]
    A rule schema $\mathcal{R}$ is an ordered triple $\langle\mathbf{A},\mathbf{S},\varphi\rangle$ where $\mathbf{A}$ is a (possibly empty) multiset of additive boxes, $\mathbf{S}$ is a single additive box, called a modal box, and $\varphi$ is a formula. All formulae in an rule schema are interpreted as schemas.
\end{definition}

We now match the individual figures of Figure~\ref{fig:NatDedILL2} with their rule schemas. :
\begin{itemize}
    \item $\irn\mto$ is written as $\langle\openaddrule\varphi\seq\psi\closeaddrule,\emptymultiset,\varphi\mto\psi\rangle$
    \item $\ern\mto$ is written as $\langle\openaddrule\seq\varphi\mto\psi\closeaddrule\msetsum\openaddrule\seq\varphi\closeaddrule,\emptymultiset,\psi\rangle$
    \item $\irn\mand$ is written as $\langle\openaddrule\seq\varphi\closeaddrule\msetsum\openaddrule\seq\psi\closeaddrule,\emptymultiset,\varphi\mand\psi\rangle$
    \item $\ern\mand$ is written as $\langle\openaddrule\seq\varphi\mand\psi\closeaddrule\msetsum\openaddrule\varphi\msetsum\psi\seq\chi\closeaddrule,\emptymultiset,\chi\rangle$
    \item $\irn\mtop$ is written as $\langle\emptymultiset,\emptymultiset,\mtop\rangle$
    \item $\ern\mtop$ is written as $\langle\openaddrule\seq\mtop\closeaddrule\msetsum\openaddrule\seq\chi\closeaddrule,\emptymultiset,\chi\rangle$
    \item $\irn\aand$ is written as $\langle\openaddrule\seq\varphi\msetsum\,\seq\psi\closeaddrule,\emptymultiset,\varphi\aand\psi\rangle$
    \item $\ern\aand$ is written as $\langle\openaddrule\seq\varphi\aand\psi\closeaddrule,\emptymultiset,\varphi\rangle$ and $\langle\openaddrule\seq\varphi\aand\psi\closeaddrule,\emptymultiset,\psi\rangle$ 
    \item Both $\irn\aor$ rules are written as $\langle\openaddrule\seq\varphi\closeaddrule,\emptymultiset,\varphi\aor\psi\rangle$ and $\langle\openaddrule\seq\psi\closeaddrule,\emptymultiset,\varphi\aor\psi\rangle$
    \item $\ern\aor$ is written as $\langle\openaddrule\seq\varphi\aor\psi\closeaddrule\msetsum\openaddrule\varphi\seq\chi\msetsum\psi\seq\chi\closeaddrule,\emptymultiset,\chi\rangle$
    \item $\irn\top$ is written as $\langle\openaddrule\emptymultiset\closeaddrule,\emptymultiset,\top\rangle$, for all $n\geq0$
    \item $\ern\abot$ is written as $\langle\openaddrule\emptymultiset\closeaddrule\msetsum\openaddrule\seq\abot\closeaddrule,\emptymultiset,\chi\rangle$, for all $n\geq0$
    \item $\promotion$ is written as $\langle\emptymultiset,\seq\varphi,\bang\varphi\rangle$
    \item $\dereliction$ is written as $\langle\openaddrule\seq\bang\varphi\closeaddrule\msetsum\openaddrule\varphi\seq\psi\closeaddrule,\emptymultiset,\psi\rangle$
    \item $\weakening$ is written as $\langle\openaddrule\seq\bang\varphi\closeaddrule\msetsum\openaddrule\seq\psi\closeaddrule,\emptymultiset,\psi\rangle$
    \item $\contraction$ is written as $\langle\openaddrule\seq\bang\varphi\closeaddrule\msetsum\openaddrule\bang\varphi\msetsum\bang\varphi\seq\psi\closeaddrule,\emptymultiset,\psi\rangle$
\end{itemize}

We call the set of these rule schemas, $\mathfrak{N}$.

\begin{definition}[Alternative intuitionistic linear derivability]\label{def:ILL-derivability-alt}
    We define a relation of derivability $\provesILL^*$ on sequents of ILL inductively according to the following two clauses:
    \begin{description}
        \item[Ref] $\varphi\provesILL^*\varphi$, for any $\varphi$.
        \item[App] Given $\langle\mathbf{A},\mathbf{S},\varphi\rangle\in\mathfrak{N}$ where $|\mathbf{A}|=m$, a (possibly empty) multiset $\bang\Delta$ where $|\bang\Delta| = k$, and $n=m+k$ (possibly empty) multisets of ILL formulae, $\Gamma_i$, such that:
        \begin{itemize}
            \item For all $\mathbf{T}_i\in\mathbf{A}$ and each $\Psi\seq\psi\in\mathbf{T}_i$ we have that $\Gamma_i\msetsum\Psi\provesILL^*\psi$, for any instantiation of $\psi$ and $\Psi$, 
            \item For each $\bang\delta_i\in \bang\Delta$ we have $\Gamma_{m+i}\provesILL^*\bang\delta_i$,
            \item For all $\Theta\seq\theta\in\mathbf{S}$ we have that $\bang\Delta\msetsum\Theta\provesILL^*\theta$, for any instantiation of $\theta$ and $\Theta$.
        \end{itemize}
         Then, $\Gamma_1\msetsum\dots\msetsum\Gamma_n\provesILL^*\varphi$.
    \end{description}
\end{definition}
It is important to note the presence of the empty additive boxes in the rules $\irn\top$ and $\ern\abot$, something seemingly completely new and very different from the previous presentation of a natural deduction system for \ILL{} in Definition~\ref{def:ILL-derivability}. In the presence of the (App) clause, we observe that these empty additive boxes are nothing more than a formalisation of the requirement that we have in Definition~\ref{def:ILL-derivability}, that the conclusion of $\irn\top$ and $\ern\abot$ follows from arbitrary many formulae in disjoint contexts. In other words, we see that this means nothing more than having an arbitrary multiset at the hypothesis line, something we have previously discussed as undesirable for our requirements, but something that exists in the literature, as in the case of the natural deduction figures of Negri for these rules in~\cite{Negri_NormalisingNatDedForILL_2002}. By having empty additive boxes instead, the rules now express that there are no conditions on what needs to hold for the conclusion of the rule to hold in the presence of an arbitrary multiset of formulae, without specifying the multiset itself. This, in the opinion of the author, is much neater than the presentation of Definition~\ref{def:ILL-derivability} and removes the question of what, in fact, is an instance of the rules $\irn\top$ and $\ern\abot$. This also makes abundantly clear an important point: The context in which the conclusion of a rule follows from in \ILL{} is connected to the additive boxes above the hypothesis line in the rule and \emph{not} the branches themselves. As previously mentioned, it is really important to note that this notion of derivability is still nothing more than what was available before.

\begin{theorem}\label{thm:ILL-derivability-equivalence}
    Let $\Gamma\seq\varphi$ be a sequent. Then it holds that $\Gamma\provesILL\varphi$ if and only if $\Gamma\provesILL^*\varphi$.
\end{theorem}
\begin{proof}
    We show this by induction over the structure of proofs. We give two examples; one where if $\Gamma\provesILL\varphi$ holds by promotion, then $\Gamma\provesILL^*\varphi$ and vice-versa, and the other where if $\Gamma\provesILL\varphi$ holds by $\irn\top$ then $\Gamma\provesILL^*\varphi$ and vice-versa. All other cases follow similarly.
    Starting from the case of promotion.
    \begin{itemize}
        \item Going left to right, we suppose $\Gamma\provesILL\varphi$ by $\promotion$. As a result, $\varphi = \bang\alpha$, and for some $n\geq0$ we have that $\Gamma=\Gamma_1\msetsum\dots\msetsum\Gamma_n$ and $\Gamma_1\provesILL\bang\psi_1$ and $\dots$ and $\Gamma_n\provesILL\bang\psi_n$ and that $\bang\psi_1\msetsum\dots\msetsum\bang\psi_n\provesILL\alpha$. By the inductive hypothesis, we have that $\Gamma_1\provesILL^*\bang\psi_1$ and $\dots$ and $\Gamma_n\provesILL^*\bang\psi_n$ and that $\bang\psi_1\msetsum\dots\msetsum\bang\psi_n\provesILL^*\alpha$. Thus, we can use the (App) clause to derive $\Gamma_1\msetsum\dots\msetsum\Gamma_n\provesILL^*\bang\alpha$, as required.
        \item Going right to left, we suppose $\Gamma\provesILL^*\varphi$ holds by (App) using the $\promotion$ rule. Recall that $\promotion$ says that $\langle\emptymultiset,[\seq\alpha],\bang\alpha\rangle$. Thus, $\varphi=\bang\alpha$ and we have, for some $n\geq 0$, a partition of $\Gamma=\Gamma_1\msetsum\dots\msetsum\Gamma_n$ and a (possibly empty) multiset $\bang\Psi =\{\bang\psi_1,\dots,\bang\psi_n\}$ such that $\Gamma_i\provesILL^*\bang\psi_i$ and such that $\bang\psi_1,\dots,\bang\psi_n\provesILL^*\alpha$. By the inductive hypothesis, we therefore have that $\Gamma_i\provesILL\bang\psi_i$ and such that $\bang\psi_1,\dots,\bang\psi_n\provesILL\alpha$. Thus, we apply $\promotion$ to obtain $\Gamma_1\msetsum\dots\msetsum\Gamma_n\provesILL\bang\alpha$, as required.
    \end{itemize}
    Now let us consider the case of $\irn\top$.
    \begin{itemize}
        \item Going left to right, we have that $\Gamma\provesILL\varphi$ by $\irn\top$. Thus, $\varphi=\top$ and for some arbitrary $n\geq0$, we have that $\Gamma_i\provesILL\psi_i$ for $i\in\{0,n\}$ and for some arbitrary formulae $\psi_1,\dots,\psi_n$ and partition of $\Gamma$ into $n$ multisets $\Gamma_1,\dots,\Gamma_n$. We are left to show that $\Gamma\provesILL^*\varphi$. By the (App) clause with the rule $\irn\top$, it follows that $\top$ follows from any arbitrary multiset of formulae vacuously. Thus $\Gamma\provesILL^*\top$ holds, as required.
        \item Going right to left, we have that $\Gamma\provesILL^*\top$ holds by $\irn\top$. We want to show that $\Gamma\provesILL\top$. Let $\Gamma=\{\gamma_1,\dots,\gamma_n\}$. Since $\gamma_i\provesILL\gamma_i$, we thust have, by $\irn\top$, that $\gamma_1\msetsum\dots\msetsum\gamma_n\provesILL\top$, i.e. $\Gamma\provesILL\top$, as required.
        
    \end{itemize}
    
\end{proof}

Thus, we see that what we have introduced with Figure~\ref{fig:NatDedILL2}, under the interpretation of Definition~\ref{def:ILL-derivability-alt}, is really nothing more than a more rigourous and uniform representation of the rules of $\rm N_{ILL}$, under the interpretation of Definition~\ref{def:ILL-derivability}. The underlying notion of derivability in the two systems remains exactly the same. However, the new notation allows us to write our inference figures in such a way that they are totally abstracted from the derivations in which they will be used, something which will be important in the remainder of this paper. As a result, for the remainder of this paper, when talking of individual inference rules (or rule schemes in the context of $\rm N_{ILL}$) we will use our Gentzen-Prawitz tree notation extended with additive and modal boxes as in Figure~\ref{fig:NatDedILL2} and stick to using sequent style inference figures to represent actual rule applications and derivations as a whole. 
\par
Before moving on, we make clear the point that, henceforth, when we write $\provesILL$, we mean the derivability relation that has hitherto been written as $\provesILL^*$. We do so, as the two notions of derivability discussed in this section are, as a consequence of Theorem~\ref{thm:ILL-derivability-equivalence}, equivalent. As a result of this, we choose the new derivability relation to be the ``standard'' for the rest of the paper, for the reasons mentioned previously but also, as the notion of atomic derivability we will introduce in the next section will be, purposefully, very similar to Definition~\ref{def:ILL-derivability-alt}. As mentioned previously, a consequence of this is that completeness will be much easier to prove.

\section{Substructural Basic Derivability}\label{sec:Basic-Derivability}
We now begin our discussion of the semantics we will develop in this paper for \ILL{}. We start by introducing the notions of an atomic rule and basic derivability which will form the core of our semantic theory.

\begin{defn}[Atomic sequent]
    An atomic sequent is an ordered pair $\langle P, p \rangle$. This ordered pair is conventionally written as $P\seq p$.
\end{defn}

\begin{defn}[Atomic additive box]
    An atomic additive box (or just additive box when the context is clear) is a (possibly empty) multiset of atomic sequents. 
\end{defn}
We use the notation $\openaddrule P\seq p, Q\seq q \closeaddrule$ to mean the atomic additive box with the two atomic sequents $P\seq p$ and $Q\seq q$. The length of an  atomic additive box is understood to mean the number of elements it contains. We conventionally write this as $l$, or possibly $l_{\mathbf{S}}$, if the atomic additive box is called $\mathbf{S}$. 

\begin{defn}[Atomic rule]\label{def:atomic-rule}
    An atomic rule is an ordered triple $\langle \mathbf{A}, \mathbf{S}, p\rangle$ where $\mathbf{A}$ is a (possibly empty) multiset of atomix additive boxes, $\mathbf{S}$ is an atomic additive box and $p$ is an atom. 
\end{defn}

Atomic rules are meant to be structurally very similar to the rule schemas of natural deduction. We shall sometimes write them graphically in a similar way too, using a tree-like notion to represent the rules. We do so as follows:\\ 
Given the atomic rule $\mathcal{R}=\langle\mathbf{A},\mathbf{S},p\rangle$, where $\mathbf{A}=\{\openaddrule Q^1_1\seq q^1_1, \dots, Q^1_{l_1}\seq q^1_{l_1}\closeaddrule,\dots,\openaddrule Q^1_1\seq q^1_1, \dots, Q^n_{l_n} \seq q^n_{l_n} \closeaddrule\}$ and $\mathbf{S} = \{U_1\seq v_1, \dots, U_m\seq v_m\}$, then the graphical form of this rule is given as:
\vspace{2mm}
\small
\small
\[
\infer[\mathcal{R}]{p}{\left\{\raisebox{-0.9em}{ \deduce{q^1_1}{[Q^1_1]}\,\dots \, \deduce{q^1_{l_1}}{[Q^1_{l_1}]}} \right\}\raisebox{-0.9em}{\,\dots\,} \left\{\raisebox{-0.9em}{ \deduce{q^n_1}{[Q^n_1]} \,\dots \, \deduce{q^n_{l_n}}{[Q^n_{l_n}]}} \right\} \quad\left\llbracket\raisebox{-0.9em}{\deduce{v_1}{[U_1]}\,\dots\,\deduce{v_m}{[U_m]}}\right\rrbracket}
\]
\normalsize

\begin{defn}[Base]\label{def:base}
    A \emph{base} is a set of atomic rules. 
\end{defn}

\begin{defn}[Persistent atom]\label{def:persistent-atom}
    An atom $p$ is said to be persistent in the base $\baseB$, if there exists a rule $\langle \emptymultiset, \mathbf{S}, p\rangle$ in $\baseB$ with non-empty $\mathbf{S}$. 
\end{defn}

When the base in question is clear, we will simply call such atoms persistent. We see that for an atom to be deemed persistent, the base in question must contain a rule whose shape closely mirrors the promotion rule in $\mathfrak{N}$. In fact, the following definition shows us that, relative to our notion of derivability in a base, persistent atoms play a role that is very similar to the role played by formulae of the shape $\bang\varphi$ in the definition of $\provesILL$.  This correspondence is important and we will use this fact in the proof of completeness of the semantics in Section~\ref{sec:Completeness}. 

\begin{definition}[Derivability in a base]~\label{def:derivability-base-ILL}
    The relation of derivability in a base $\baseB$, denoted as $\deriveBaseM{\baseB}$, is a relation, indexed by the base $\baseB$, on atomic sequents, defined inductively according to the following two clauses:
    \begin{description}
        \item[Ref\label{eq:derive-ref}] 
        $p \deriveBaseM{\baseB} p$. 
        \item[App\label{eq:derive-app}] 
        Given $\langle \mathbf{A}, \mathbf{S}, p\rangle \in \baseB$ where $|\mathbf{A}| = m$, a (possibly empty) multiset of persistent atoms, $D$, where $|D|=k$, and $n = m + k$ (possibly empty) atomic multisets $C_i$, such that:
        \begin{itemize}
            \item For each additive box $\mathbf{T}_i \in \mathbf{A}$ and each atomic sequent $Q\seq q \in \mathbf{T}_i$, we have $C_i\msetsum Q\deriveBaseM{\baseB}q$,
            \item For each $d_i\in D$, we have $C_{m+i}\deriveBaseM{\baseB}d_i$,
            \item For all $U\seq v\in \mathbf{S}$ we have $D\msetsum U\deriveBaseM{\baseB}v$
        \end{itemize}
        Then, $C_{1}\msetsum\dots\msetsum C_{n}\deriveBaseM{\baseB} p$.
        \end{description}
\end{definition}

When $L\deriveBaseM{\baseB}p$ holds, we say that $p$ follows (or is derivable) from $L$ in base $\baseB$. The multiset $L$ is called the multiset of hypotheses (or sometimes hypothesis multiset) of the derivation. 
As a result of the similarity between the notion of derivation in a base and that of $\provesILL$, we can, in fact, graphically represent derivations in a base using a sequent-like notation very similar to that shown in Figure~\ref{fig:NatDedILL}. This graphical notation proves very useful when considering explicit and long atomic derivations, though we will not be making use of it in this work. 

\begin{proposition}[Monotonicity of $\deriveBaseM{\baseB}$]\label{lem:monotone-derivability}
    If $P \deriveBaseM{\baseB} p$ then for all $\baseC \baseGeq \baseB$ we also have that $P \deriveBaseM{\baseC} p$.
\end{proposition}
\begin{proof}
    Supposing the hypothesis, then $P \deriveBaseM{\baseB} p$ holds in one of two ways.
    \begin{itemize}
        \item If $P \deriveBaseM{\baseB}p$ holds due to~\eqref{eq:derive-ref}, then $P=p$ and so it holds for any base $\baseX$ that $P \deriveBaseM{\baseX}p$.
        \item Else it must be the case that $P \deriveBaseM{\baseB}p$ holds by~\eqref{eq:derive-app}. The result follows by noting that if there are rules in $\baseB$ allowing a derivation of $p$ from $P$ then those rules will also be in $\baseC$ for all $\baseC \baseGeq \baseB$, and thus the derivation still holds.
    \end{itemize}
\end{proof}

\begin{lemma}[Cut admissibility for $\deriveBaseM{\baseB}$]
\label{lem:atomic-cut}
    The following are equivalent for arbitrary atomic multisets $P\msetsum  S$, atom $q$, and base $\baseB$, where we assume $P = \makeMultiset{p_1,\dots, p_n}$: 
    \begin{enumerate}
        \item $P\msetsum  S \deriveBaseM{\baseB} q$.~\label{eq:atomic-cut-1}
        \item For every $\baseC \baseGeq \baseB$, atomic multisets $T_1, \dots, T_n$ where $T_1 \deriveBaseM{\baseC} p_1, \dots, T_n \deriveBaseM{\baseC} p_n$, then $T_1\msetsum\dots\msetsum T_n\msetsum S \deriveBaseM{\baseC} q$.~\label{eq:atomic-cut-2}
    \end{enumerate}
\end{lemma}

\begin{proof}
    We begin by proving that~\eqref{eq:atomic-cut-2} implies~\eqref{eq:atomic-cut-1}. For this, we begin by taking $\baseC = \baseB$ and $T_i$ to be $\makeMultiset{p_i}$ for each $i = 1, \dots, n$. Since $p_1 \deriveBaseM{\baseB} p_1, \dots, p_n \deriveBaseM{\baseB} p_n$ all hold by~\eqref{eq:derive-ref}, it thus follows from~\eqref{eq:atomic-cut-2} that $p_1\msetsum\dots\msetsum p_n\msetsum S \deriveBaseM{\baseB} q$ which is nothing more than $P\msetsum S \deriveBaseM{\baseB} q$.
    
    To now show~\eqref{eq:atomic-cut-1} implies~\eqref{eq:atomic-cut-2}, we need to consider how $P\msetsum S \deriveBaseM{\baseB} q$ is derived; that is, we proceed by induction, considering the cases when the derivation holds due to~\eqref{eq:derive-ref}, our base case, and~\eqref{eq:derive-app} separately.
    \begin{itemize}[label={-}]
        \item $P\msetsum S \deriveBaseM{\baseB} q$ holds by~\eqref{eq:derive-ref}. This gives us that $P\msetsum S = \makeMultiset{q}$, giving $q \deriveBaseM{\baseB} q$. There are thus two cases to consider, depending on which of $P$ and $S$ is $\makeMultiset{q}$. 
        \begin{itemize}
        \item[Case 1:] $P = \makeMultiset{q}$ and $S = \emptymultiset$. In this case, the statement of~\eqref{eq:atomic-cut-2} becomes, for every $\baseC \baseGeq \baseB$ and atomic multiset $T$ where $T \deriveBaseM{\baseC} q$, then $T \deriveBaseM{\baseC} q$. This holds trivially.
        \item[Case 2:] $P = \emptymultiset$ and $S = \makeMultiset{q}$. In this case, the statement of~\eqref{eq:atomic-cut-2} becomes, for every $\baseC \baseGeq \baseB$, $S \deriveBaseM{\baseC} q$. This holds by hypothesis from~\eqref{eq:atomic-cut-1}.
        \end{itemize}
        We are now left to show that~\eqref{eq:atomic-cut-1} implies~\eqref{eq:atomic-cut-2} according to~\eqref{eq:derive-app}. We show this by induction on the structure of the derivation $P\msetsum S\deriveBaseM{\baseB}q$.
        \item $P\msetsum S \deriveBaseM{\baseB} q$ holds by~\eqref{eq:derive-app}. \\
        Start by supposing that the rule we apply~\eqref{eq:derive-app} with is $\langle\mathbf{A},\mathbf{S},q\rangle \in \baseB$, where the size of $\mathbf{A}$ is $m\leq n$. Thus, we must have partitions of $P$ and $S$ into $P=P_1\msetsum\dots\msetsum P_{n}$ and $S=S_1\msetsum\dots\msetsum S_{n}$ such that:
        \begin{itemize}
            \item We have some multiset of persistent atoms $D = \{d_{m+1},\dots,d_{n}\}$ such that the derivations $P_{m+i}\msetsum S_{m+i}\deriveBaseM{\baseB}d_{m+i}$ hold for $i\in[1,n-m]$ and that for each atomic sequent $U\seq v \in \mathbf{S}$ we have that $D\msetsum U\deriveBaseM{\baseB} v$ holds.
            \item For each $\mathbf{T}_i \in \mathbf{A}$ and $Q\seq r \in \mathbf{T}_i$, we have that $P_i\msetsum S_i\msetsum Q \deriveBaseM{\baseB}r$
        \end{itemize}
        The hypothesis gives that we have multisets $T_1, \dots, T_n$ such that $T_1 \deriveBaseM{\baseC} p_1, \dots, T_n \deriveBaseM{\baseC} p_n$ hold. Since each $P_i$ is a partition of $P$, we know that they can be written as $P_i=\makeMultiset{p_{i_1},\dots,p_{i_{l_i}}}$. 
        Therefore, we can similarly partition each $T_i$ such that $T_i = T_{i_1}\msetsum\dots\msetsum T_{i_{l_i}}$. 
        By the inductive hypothesis, we therefore have that:
        \begin{itemize}
            \item We have some multiset of persistent atoms $D = \{d_{m+1},\dots,d_{n}\}$ such that the derivations $T_{(m+i)_1}\msetsum\dots\msetsum T_{(m+i)_{l_{(m+i)}}}\msetsum S_{m+i}\deriveBaseM{\baseB}d_{m+i}$ hold for $i\in[1,n-m]$ and that for each atomic sequent $U\seq v \in \mathbf{S}$ we have that $D\msetsum U\deriveBaseM{\baseB} v$ continue to hold.
            \item For each $\mathbf{T}_i \in \mathbf{A}$ and $Q\seq r \in \mathbf{T}_i$, that $T_{i_1}\msetsum\dots\msetsum T_{i_{l_i}}\msetsum S_i\msetsum Q\deriveBaseM{\baseC}r$
        \end{itemize}
        Since $\langle\mathbf{A},\mathbf{S},q\rangle \in \baseB$ and $\baseB \baseLeq \baseC$, then, by Lemma~\ref{lem:monotone-derivability} and~\eqref{eq:derive-app}, it follows that $T_{1_1}\msetsum\dots\msetsum T_{1_{l_1}}\msetsum S_1\msetsum \dots\msetsum T_{n_1}\msetsum\dots\msetsum T_{n_{l_n}}\msetsum S_n \deriveBaseM{\baseC}q$, which, when rearranged, gives $T_1\msetsum\dots\msetsum T_n\msetsum S_1\msetsum\dots\msetsum S_n \deriveBaseM{\baseC}q$, which is nothing more than $T_1\msetsum\dots\msetsum T_n\msetsum S \deriveBaseM{\baseC}q$, completing the induction as required.
    \end{itemize}
\end{proof}

\section{Base-extention Semantics}\label{sec:BeS}
We are now ready to introduce the support relation, as mentioned in the introduction, the relation at the heart of the semantics we define in this paper
.
\begin{definition}[Support]\label{def:support}
    The relation of support, denoted as $\suppM{\baseB}{L}$, is a relation on sequents, indexed by a base $\baseB$ and a (finite) atomic multiset $L$, defined inductively according to the definitions of Figure~\ref{fig:ILL:support}. Note that $\Gamma,\Delta$ and $\Theta$ are non-empty multisets. Furthermore, the multiset $\Theta$ contains no formulae with $(\bang)$ as a top-level connective.
\begin{figure}[th]
    \hrule \vspace{1mm}
              \[
            \begin{array}{r@{\qquad}l@{\quad}c@{\quad}l}
               \mbox{(At)} & \suppM{\baseB}{L} p  & \text{ iff } &   L\deriveBaseM{\baseB} p  \\ [1mm]
                (\mto) & \suppM{\baseB}{L} \varphi \mto \psi & \text{ iff } & \varphi \suppM{\baseB}{L} \psi \\[1mm]
                (\mand) & \suppM{\baseB}{L} \varphi \mand \psi   & \text{ iff } &  \text{for any } \baseC  \text{ such that } \baseC \baseGeq \baseB \text{, atomic multisets } K \\ 
                & & & \text{and any } p \in \At, \text{ if } \varphi\msetsum\psi\suppM{\baseB}{K}p \text{ then } \suppM{\baseC}{L\msetsum K} p  \\[1mm]
                (\mtop) & \suppM{\baseB}{L} \mtop   & \text{ iff } &  \text{for any } \baseC  \text{ such that } \baseC \baseGeq \baseB \text{, atomic multisets } K \\ 
                & & & \text{and any } p \in \At, \text{ if } \suppM{\baseB}{K}p \text{ then } \suppM{\baseC}{L\msetsum K} p  \\[1mm]
                (\aand) & \suppM{\baseB}{L} \varphi \aand \psi   & \text{ iff } &   \suppM{\baseB}{L} \varphi \text{ and }   \suppM{\baseB}{L} \psi  \\[1mm] 
                (\aor) & \suppM{\baseB}{L} \varphi \aor \psi & \text{ iff } &  \text{for any } \baseC  \text{ such that } \baseC \baseGeq \baseB \text{, atomic multisets } K \\ 
                & & & \text{and any } p \in \At, \text{ if } \varphi \suppM{\baseC}{K} p \text{ and } \psi \suppM{\baseC}{K} p, \text{ then } \suppM{\baseC}{L\msetsum K} p  \\[1mm]
                (\abot) & \suppM{\baseB}{L} \abot & \text{iff} &    \suppM{\baseB}{L\msetsum K} p \text{ for any } p \in \At \text{ and } K\subset\At \\[1mm]
                (\top) & \suppM{\baseB}{L}\top & & \text{always}\\[1mm]
                (\bang) & \suppM{\baseB}{L} \bang \varphi & \text{iff} &  \text{for any } \baseC  \text{ such that } \baseC \baseGeq \baseB \text{, atomic multisets } K \\ 
                & & & \text{and any } p \in \At, \text{ if for any } \baseD \text{ such that } \baseD\baseGeq\baseC, \\
                & & & (\text{if } \suppM{\baseD}{\emptymultiset}\varphi \text{ then } \suppM{\baseD}{L}p) \text{ then } \suppM{\baseC}{L\msetsum K}p  \\[1mm]
                (\msetsum) & \suppM{\baseB}{L}\Gamma\msetsum\Delta & \text{iff} & \text{there exists multisets } K \text{ and } M \text{ such that } L=K\msetsum M \\ 
                & & &\text{and } \suppM{\baseB}{K}\Gamma \text{ and } \suppM{\baseB}{M}\Delta \\
                \mbox{(Inf)} & \hspace{-1em} \bang\Delta\msetsum\Theta \suppM{\baseB}{L} \varphi & \text{ iff } & \text{for any } \baseC \text{ such that } \baseC \baseGeq \baseB \text{, atomic multisets } K, \\
                & & & \text{if } \suppM{\baseC}{\emptymultiset}\Delta \text{ and } \suppM{\baseC}{K} \Theta \text{ then } \suppM{\baseC}{L\msetsum K} \varphi \\[1mm]
            \end{array}
            \]
    \hrule
    \vspace{1mm}
        \caption{Support for Intuitionistic Linear Logic}
       ~\label{fig:ILL:support}
        \vspace{-15pt}
       ~\labelandtag{BeS:ILL:at}{(At)}
       ~\labelandtag{BeS:ILL:mto}{$(\mto)$}
       ~\labelandtag{BeS:ILL:mand}{$(\mand)$}
       ~\labelandtag{BeS:ILL:mtop}{$(\mtop)$}
       ~\labelandtag{BeS:ILL:aand}{$(\aand)$}
       ~\labelandtag{BeS:ILL:aor}{$(\aor)$}
       ~\labelandtag{BeS:ILL:abot}{$(\abot)$}
       ~\labelandtag{BeS:ILL:atop}{$(\top)$}
       ~\labelandtag{BeS:ILL:bang}{$(\bang)$}
       ~\labelandtag{BeS:ILL:comma}{$(\msetsum)$}
       ~\labelandtag{BeS:ILL:inf}{(Inf)}
\end{figure}
\end{definition}

\begin{definition}[Validity]\label{def:ILL-validity}
    The sequent $(\Gamma:\phi)$ is said to be valid if and only if for all bases $\baseB$, it is the case that $\Gamma\suppM{\baseB}{\emptymultiset}\varphi$.
\end{definition}

That the inductive definition of $\suppM{\baseB}{L}$ is well-founded may not be immediately clear. To show this, we define the following notion of the degree of a formula:
\begin{itemize}
    \item To atoms $p$, we assign a degree of $1$.
    \item To the constants $\top$, $\mtop$ and $\abot$ assign a degree of $2$.
    \item To each formula $\varphi\mto\psi$, $\varphi\mand\psi$, $\varphi\aand\psi$ and $\varphi\aor\psi$, assign the degree the sum of the degrees of $\varphi$ and $\psi$ plus $1$.
    \item To $\bang\varphi$, assign the degree of $\varphi$ plus $1$.
\end{itemize}
For all the definitional clauses in Definition~\ref{def:support} we have that the formula being defined is always of \emph{greater} degree than any formula in its definition, thus verifying the claim. 

\begin{lemma}\label{lem:atomic-sound-and-complete}
    $L \suppM{\baseB}{K} p$ iff $L\msetsum K \deriveBaseM{\baseB} p$.
\end{lemma}
\begin{proof}
    If $L=\emptymultiset$, then the result holds immediately by~\ref{BeS:ILL:at}. So consider $L = \makeMultiset{l_1,\dots,l_n}$. Proceeding from right to left, we begin by immediately applying Lemma~\ref{lem:atomic-cut} to the hypothesis which gives us that for all $\baseC \baseGeq \baseB$ and atomic multisets $T_i$ such that $T_i\deriveBaseM{\baseC}l_i$ for $i=1,\dots,n$, that we have $T_1\msetsum\dots\msetsum T_n\msetsum K\deriveBaseM{\baseC}p$. By~\ref{BeS:ILL:at}, this means that $\suppM{\baseC}{T_1\msetsum\dots\msetsum T_n\msetsum K}p$. Since we have that $T_i\deriveBaseM{\baseC}l_i$ for $i=1,\dots,n$ then by~\ref{BeS:ILL:at} we therefore have that $\suppM{\baseC}{T_i}l_i$ for $i=1,\dots,n$. Thus by~\ref{BeS:ILL:inf} we conclude that $L\suppM{\baseB}{K}p$. Finally, because~\ref{BeS:ILL:inf},~\ref{BeS:ILL:at} and Lemma~\ref{lem:atomic-cut} are bi-implications, that therefore completes the proof.
\end{proof}

An important consequence of this theorem is that $\suppM{\baseB}{L}$ makes for a conservative extension of $L\deriveBaseM{\baseB}$ to the whole language of Intuitionistic linear logic (that is, $\suppM{\baseB}{L}p$ if and only if $L\deriveBaseM{\baseB}p$). As a result, it \emph{should} be the case that we retain montonicity in the base, as follows.

\begin{lemma}[Monotonicity of $\suppM{\baseB}{L}$]\label{lem:monotone-support}
    If $\Gamma \suppM{\baseB}{L} \varphi$ then for all $\baseC \baseGeq \baseB$, we have that $\Gamma \suppM{\baseC}{L} \varphi$ holds.
\end{lemma}

The proof follows immediately from Lemma~\ref{lem:monotone-derivability} and the inductive clauses of Definition~\ref{def:support}. Note however, that the support relation is monotone \emph{only} with respect to the base, not with respect to the context. For example, we see that $\suppM{\baseB}{p}p$ holds in all bases $\baseB$, but $\suppM{\baseB}{p\msetsum p} p$ does not necessarily hold in all bases $\baseB$. Were it to be the case that the support relation was monotone with respect to the context as well it would result in us having unrestricted weakening and contraction in the context; something which, as we shall see, is undesirable for the semantics we have set up. The fact that that the support relation is monotone with respect to the base, however, is useful and in fact, allows us to give us a simpler characterisation of validity.
\begin{lemma}\label{lem:ILL-validity}
    The sequent $(\Gamma:\phi)$ is valid if and only if $\Gamma\suppM{\emptybase}{\emptymultiset}\phi$.
\end{lemma}
\begin{proof}
    Going left to right, we have that for all bases $\baseB$, it is the case that $\Gamma\suppM{\baseB}{\emptymultiset}\phi$. Thus, in particular, we have that $\Gamma\suppM{\emptybase}{\emptymultiset}\varphi$. Going right to left, the result follows by monotonicity as the empty base is the smallest subset of every base.
\end{proof}
Thus, we can justifiably write the valid sequent $(\Gamma:\phi)$ as $\Gamma\suppM{}{}\phi$. Before moving on, it is worth nothing that the clause for ($\bang$), can be simplified to
\begin{align*}
    \suppM{\baseB}{L}\bang \varphi&\text{ iff for all bases }\baseC\baseGeq\baseB\text{, atomic multisets } K \text{ and } p\in\At,\\
    &\text{ if }\bang\varphi\suppM{\baseB}{K}p\text{ then } \suppM{\baseB}{L\msetsum K}p
\end{align*}
\noindent That this holds is an immediate consequence of our (Inf) clause and we make use of this form of the definition for the remainder of this paper. We now make note of three interesting structural properties of $\suppM{\baseB}{L}$, the proofs of which we defer to Appendix~\ref{secA:Proofs}.
\begin{lemma}\label{lem:mand-key-lemma}
    Given $\suppM{\baseB}{L} \varphi \mand \psi$ and $\varphi\msetsum\psi \suppM{\baseB}{K} \chi$ then $\suppM{\baseB}{L\msetsum K} \chi$ holds.
\end{lemma}
\begin{lemma}\label{lem:mtop-key-lemma}
    Given $\suppM{\baseB}{L} \mtop$ and $\suppM{\baseB}{K} \chi$ then $\suppM{\baseB}{L\msetsum K} \chi$ holds.
\end{lemma}
\begin{lemma}\label{lem:aor-key-lemma}
    Given that $\suppM{\baseB}{L} \varphi \aor \psi$, $\varphi \suppM{\baseB}{K} \chi$ and $\psi \suppM{\baseB}{K} \chi$ all hold then $\suppM{\baseB}{L\msetsum K} \chi$.
\end{lemma}

Of more interest to us is the interaction of $\suppM{\baseB}{L}$ with formulae $(\bang)$ as a top-level connective. Let us prove some structural results about these formulae.

\begin{lemma}\label{lem:bang-dereliction}
    Given $\varphi \suppM{\baseB}{L}\psi$ then $\bang\varphi\suppM{\baseB}{L}\psi$.
\end{lemma}
\begin{proof}
    Begin by fixing some base $\baseC\baseGeq\baseB$ such that $\suppM{\baseC}{\emptymultiset}\varphi$. We wish to show that $\suppM{\baseC}{L}\psi$. The given hypothesis is equivalent to the statement that $\suppM{\baseX}{K}\varphi$ implies $\suppM{\baseX}{L\msetsum K}\psi$ for all bases $\baseX\baseGeq\baseB$ and atomic multisets $K$. In particular, we consider when $K=\emptymultiset$ and $\baseX=\baseC$, at which point we conclude $\suppM{\baseC}{L}\psi$, as required.
\end{proof}

\begin{lemma}\label{lem:bang-necessitation}
    Given $\suppM{\baseB}{\emptymultiset}\varphi$ and for all $\baseC\baseGeq\baseB$ such that $\bang\varphi\suppM{\baseC}{L}\psi$, then $\suppM{\baseC}{L}\psi$.
\end{lemma}
\begin{proof}
    We begin by fixing arbitrary $\baseC$ such that $\bang\varphi\suppM{\baseC}{L}\psi$. By (Inf), we have that $\bang\varphi\suppM{\baseC}{L}\psi$ is equivalent to the statement that $\suppM{\baseX}{\emptymultiset}\varphi$ implies $\suppM{\baseX}{L}\psi$, for all bases $\baseX\baseGeq\baseC$. In particular, we consider when $\baseX=\baseC$. Since we have that $\suppM{\baseB}{\emptymultiset}\varphi$, then by Lemma~\ref{lem:monotone-support}, we have that $\suppM{\baseX}{\emptymultiset}\varphi$ and thus $\suppM{\baseC}{L}\psi$.
\end{proof}
\begin{corollary}\label{cor:bang-pure-necessitation}
    Given $\suppM{\baseB}{\emptymultiset}\varphi$ then $\suppM{\baseB}{\emptymultiset}\bang\varphi$.
\end{corollary}
\begin{proof}
    We start by fixing a base $\baseC\baseGeq\baseB$, atomic multiset $K$ and an atom $p$ such that $\bang\varphi\suppM{\baseC}{K}p$. We are left to show $\suppM{\baseC}{K}p$, which follows immediately by Lemma~\ref{lem:bang-necessitation}, with $\psi=p$.
\end{proof}

\begin{corollary}\label{cor:bang-promotion}
    Given $\bang \Gamma \suppM{}{} \varphi$ then $\bang \Gamma \suppM{}{} \bang \varphi$.
\end{corollary}
\begin{proof}
    We start by considering all bases $\baseB$ where $\bang\Gamma$ is supported. Thus, our hypothesis becomes, under the quantifier, that $\suppM{\baseB}{\emptymultiset}\varphi$. By Corollary~\ref{cor:bang-pure-necessitation}, we thus have, under the quantifier, $\suppM{\baseB}{\emptymultiset}\bang\varphi$ and there fore, $\bang\Gamma\suppM{}{}\bang\varphi$, as required.
\end{proof}

The question ``Does the deduction theorem hold in modal logics?" has long been a problematic issue in the literature of modal logic~\cite{HakliNegri_DoesDeductionTheoremFailForModalLogics_2012}. These historical issues seem to be being reflected in our definition of~\ref{BeS:ILL:inf}, which seems to suggest, somewhat in line with the reasoning of Fagin et al. in~\cite{book_Reasoning_about_knowledge} for Epistemic Logic and Chagrov and Zakharyaschev~\cite{book_Modal_Logic} for normal modal logics, that modal formulae (that is, for us, formulae of the form $\bang\varphi$) need to be treated differently as hypothesis. This is troublesome for the linear logician as there are many formulae which are logically equivalent to a ``modal" formula but do not have the $(\bang)$ as a top-level connective. An example of this are the formulae 
\[
    \bang(\varphi\aand\psi) \mbox{ and } (\bang\varphi)\mand(\bang\psi)
\] 
Our desire to treat formulae of the form $\bang\varphi$ differently as hypothesis comes from the observation that the sequent 
\[
    \bang\varphi\seq\varphi\mand\dots\mand\varphi
\] 
is meant to be valid for any number of $\varphi$'s in the right-hand side. Thus, it should be the case that for any number of $\varphi$, we have that 
\[
    \bang\varphi\suppM{}{}\varphi\mand\dots\mand\varphi
\] 
However, according to the resource interpretation of this support judgement, a ``normal", resourceful, reading of $\bang\varphi$ on the left-hand side would seemingly lead to a contradiction to the fact that our sequent is valid for any number of $\varphi$ on the right. Spelt out, a ``normal" style $\mbox{(Inf)}$ clause would say that 
\begin{align*}
\bang\varphi\suppM{}{}\varphi\mand\dots\mand\varphi &\mbox{ iff for all } \baseB, \mbox{ atomic multisets } L \mbox{ and atoms } p,\\
&\mbox{ if } \suppM{\baseB}{L}\bang\varphi, \mbox{ then } \suppM{\baseB}{L}\varphi\mand\dots\mand\varphi
\end{align*}
Since $L$ must be finite, it seems preposterous that such a judgement would hold for \emph{arbitrary} many $\varphi$. Nevertheless, this indeed turns out to be the case. In fact, we dedicate the remainder of this section to showing that the ``normal" style $\mbox{(Inf)}$ clause is equivalent to our~\ref{BeS:ILL:inf} clause. That is, we will prove the following statement:
\[
    \Gamma\suppM{\baseB}{L}\varphi \text{ iff for all }\baseC\baseGeq\baseB\text{ and atomic multisets }K,\, \suppM{\baseC}{K}\Gamma\text{ implies }\suppM{\baseC}{L\msetsum K}\varphi
\]\labelandtag{BeS:ILL:gen-inf}{(Gen-Inf)}
We call this clause $\mbox{(Gen-Inf)}$, in line with~\cite{AlexTaoDavid_ResourceSemantics_2024}.
Thus, the goal for the remainder of this section is to show that~\ref{BeS:ILL:gen-inf} indeed holds. An interesting point to note is that, as a result of the soundness and completeness of the semantics, one can view the proof of this statement as a semantic argument \emph{for} the deduction theorem in ILL. Before continuing, since~\ref{BeS:ILL:gen-inf} holds (and thus there is no need for a special~\ref{BeS:ILL:inf} clause), it is worth explaining the reasoning behind sticking with such an~\ref{BeS:ILL:inf} clause: Firstly, working with such a form of the~\ref{BeS:ILL:inf} clause gives us a new way of looking at an aspect of base-extension semantics that has, so far in the literature, never needed to be investigated further. Secondly, and perhaps more pragmatically, it makes the mathematics simpler. Whether making such a choice for simplicity alone is a discussion unto itself, which we shall not be continuing here. In any case, we proceed to show that~\ref{BeS:ILL:gen-inf} is indeed equivalent to~\ref{BeS:ILL:inf}.

\begin{lemma}\label{lem:bang-cut-support-prim}
    Given $\suppM{\baseB}{L}\bang \varphi$ and $\bang \varphi \suppM{\baseB}{K}\psi$ then $\suppM{\baseB}{L\msetsum K}\psi$.
\end{lemma}
\begin{proof}
    We proceed by induction on the structure of $\psi$. We show three cases, one multiplicative, one additive and the base case to highlight the different aspects of the induction.
    \begin{itemize}[label={-}]
        \item $\psi = p$ for some $p\in \At$. In this case our second hypothesis says $\bang \varphi \suppM{\baseB}{K}p$ and we want to show that $\suppM{\baseB}{L\msetsum K} p$. The first hypothesis is equivalent to the statement that for all $\baseX \baseGeq \baseB$, atomic multisets $M$ and atoms $q$, if $\bang \varphi \suppM{\baseX}{M}q$ then $\suppM{\baseX}{L\msetsum M}q$. Thus, considering when $\baseX=\baseB$, $M=K$ and $q=p$, we obtain that $\suppM{\baseB}{L\msetsum K}p$, as requred.
        \item $\psi = \alpha \mand \beta$. In this case, it is equivalent to show that from our original hypotheses and from the hypothesis that for all bases $\baseX \baseGeq \baseB$ atomic multisets $M$ and atoms $p \in \At$ that $\alpha\msetsum\beta\suppM{\baseC}{M}p$, that $\suppM{\baseC}{L\msetsum K\msetsum M}p$ holds. To do that we need to show that $\bang\varphi\suppM{\baseC}{K\msetsum M}p$. To show this, we first consider our second hypothesis $\bang\varphi\suppM{\baseC}{K}\alpha\mand\beta$, which holds in $\baseC\baseGeq\baseB$ by monotonicity. By (Inf), this is equivalent to considering for all $\baseD\baseGeq\baseC$ such that if $\suppM{\baseD}{\emptymultiset}\varphi$ then $\suppM{\baseD}{K}\alpha\mand\beta$. The conclusion here is equivalent to considering all extensions $\baseE\baseGeq\baseD$, atomic multisets $N$ and atoms $p$, if $\alpha\msetsum\beta\suppM{\baseE}{N}p$ then $\suppM{\baseE}{K\msetsum N}p$. By monotonicity, and in particular at $\baseE=\baseD$ and $N=M$ our additional hypothesis gives that $\suppM{\baseD}{K\msetsum M}p$. Thus, by (Inf), we have $\bang\varphi\suppM{\baseC}{K\msetsum M}p$, as required.
        To finish this proof off, we note that our first point, by ($\bang$) is equivalent to considering all $\baseX\baseGeq\baseB$, atomic multisets $N$ and atoms $p$, such that if $\bang\varphi\suppM{\baseX}{N}p$ then $\suppM{\baseX}{L\msetsum N}p$. In particular, when $\baseX=\baseC$ and $N=K\msetsum M$ we obtain our desired conslusion $\suppM{\baseC}{L\msetsum K\msetsum M}p$.
        \item $\psi = \alpha \aor \beta$. In this case, we take as additional hypothesis a base $\baseC\baseGeq\baseB$, atomic multiset $M$ and an atom $p$ such that $\alpha\suppM{\baseC}{M}p$ and $\beta\suppM{\baseC}{M}p$. Our goal will be to show that $\suppM{\baseC}{L\msetsum K\msetsum M}p$. We note that the first hypothesis, $\suppM{\baseB}{L}\bang\varphi$, implies that, if $\bang\varphi\suppM{\baseC}{K\msetsum M}p$ then $\suppM{\baseC}{L\msetsum K\msetsum M}p$. Thus, we show that $\bang\varphi\suppM{\baseC}{K\msetsum M}p$. By Lemma~\ref{lem:monotone-support}, the second hypothesis gives $\suppM{\baseC}{K}\alpha\aor\beta$. This is equivalent to considering all bases $\baseD\baseGeq\baseC$ and atomic multisets $N$ such that $\suppM{\baseD}{\emptymultiset}\varphi$ implies $\suppM{\baseD}{K}\alpha\aor\beta$. The conclusion of this implication is equivalent to considering all bases $\baseE\baseGeq\baseD$, atomic multisets $N$ and atoms $q$ such that $\alpha\suppM{\baseE}{N}q$ and $\beta\suppM{\baseE}{N}q$ imply $\suppM{\baseE}{K\msetsum N}q$. Since we have by hypothesis that $\alpha\suppM{\baseC}{M}p$ and $\beta\suppM{\baseC}{M}p$, by considering the case when $\baseE=\baseD$, $N=M$ and $q=p$, it therefore follows by (Inf) that $\bang\varphi\suppM{\baseC}{K\msetsum M}p$.
    \end{itemize}
    All other cases follow similarly.
\end{proof}

\begin{corollary}\label{cor:bang-cut-support}
    Given $\suppM{\baseB}{L}\bang \Gamma$ and $\bang \Gamma \suppM{\baseB}{K}\psi$ then $\suppM{\baseB}{L\msetsum K}\psi$.
\end{corollary}
This corollary is an immediate consequence of Lemma~\ref{lem:bang-cut-support-prim}.

\begin{theorem}\label{thm:cut-support}
    $\Gamma\suppM{\baseB}{L}\varphi$ if and only if for all $\baseC\baseGeq\baseB$ and atomic multisets $K$, $\suppM{\baseC}{K}\Gamma$ implies $\suppM{\baseC}{L\msetsum K}\varphi$.
\end{theorem}

\begin{proof}
    We begin by supposing we have a partition of $\Gamma$ into formulae with $(\bang)$ as a top level connective, which we denote as $\bang\Delta$, and those which don't, which we write as $\Theta$. Thus $\Gamma = \bang\Delta\msetsum\Theta$.

    Going left to right, we begin by fixing an arbitrary base $\baseC\baseGeq\baseB$ and atomic multiset $K$ such that $\suppM{\baseC}{K}\bang\Delta\msetsum\Theta$, which is to say, there is a partition of $K=M\msetsum N$ such that $\suppM{\baseC}{M}\bang\Delta$ and $\suppM{\baseC}{N}\Theta$. It now suffices to show $\suppM{\baseC}{L\msetsum M\msetsum N}\varphi$. To this end, we consider $\Gamma\suppM{\baseB}{L}\varphi$ which we now write as $\bang\Delta\msetsum\Theta\suppM{\baseB}{L}\varphi$. By~\ref{BeS:ILL:inf}, this is equivalent to considering all bases $\baseX\baseGeq\baseB$ and atomic multisets $Q$ such that if $\suppM{\baseX}{\emptymultiset}\Delta$ and $\suppM{\baseX}{Q}\Theta$ then $\suppM{\baseX}{L\msetsum Q}\varphi$ which itself is equivalent to considering all bases $\baseX\baseGeq\baseB$ and atomic multisets $Q$ such that if $\suppM{\baseX}{Q}\Theta$ then $\bang\Delta\suppM{\baseX}{L\msetsum Q}\varphi$. Considering when $\baseX=\baseC$ and $Q=N$, we obtain that $\bang\Delta\suppM{\baseC}{L\msetsum N}\varphi$. Since we also have by hypothesis that $\suppM{\baseC}{M}\bang\Delta$ then, by Corollary~\ref{cor:bang-cut-support}, we conclude $\suppM{\baseC}{L\msetsum M\msetsum N}\varphi$, as required.

    Going right to left, we start by fixing a base $\baseD\baseGeq\baseB$ and an atomic multiset $M$ such that $\suppM{\baseD}{\emptymultiset}\Delta$ and $\suppM{\baseD}{M}\Theta$. It remains to show that $\suppM{\baseD}{L\msetsum M}\varphi$. Observe that, by Corollary~\ref{cor:bang-pure-necessitation} applied to each element of $\Delta$, we have that $\suppM{\baseD}{\emptymultiset}\bang\Delta$. Thus, we have that $\suppM{\baseD}{M}\bang\Delta\msetsum\Theta$, which is to say, $\suppM{\baseD}{M}\Gamma$. By considering the given implication with $\baseC=\baseD$ and $K=M$ we therefore obtain $\suppM{\baseC}{L\msetsum M}\varphi$ as required.

\end{proof}

Finally, it is worth mentioning the following points related to the multiplicative unit. These are to be expected. We won't make much use of these results in what follows, except for the second point, as this allows us to prove soundness of the weakening rule, though they are nice sanity checks:
\begin{itemize}
    \item By expanding the definition of $\mtop$, one sees that $\suppM{\emptybase}{\emptymultiset}\mtop$ is valid. 
    \item Consequently, it is the case that for any $\varphi$, we have that $\suppM{\baseB}{L}\varphi \text{ iff } \mtop\suppM{\baseB}{L}\varphi$. Going right to left, it holds by Lemma~\ref{lem:mtop-key-lemma}. Going right to left, it holds immediately since we cut with $\suppM{\emptybase}{\emptymultiset}\mtop$.
    \item Finally, we have that if $\suppM{\baseB}{L}\mtop$ and $\suppM{\baseB}{K}\mtop$ hold, then $\suppM{\baseB}{L\msetsum K}\mtop$ also holds, again by cut.
\end{itemize}

We finish this section with following lemma which relates derivations of ($\bang$) to derivations of ($\mtop$).
\begin{lemma}\label{lem:bang-mtop-interplay}
    Given $\suppM{\baseB}{L}\bang\varphi$ then $\suppM{\baseB}{L}\mtop$.
\end{lemma}
\begin{proof}
    We start by fixing a base $\baseC\baseGeq\baseB$, atomic multiset $K$ and an atom $p$ such that $\suppM{\baseC}{K}p$ and note that must show that $\suppM{\baseC}{L\msetsum K}p$. By hypothesis, and Lemma~\ref{lem:monotone-support}, we have that $\suppM{\baseC}{L}\bang\varphi$. This is equivalent to for all $\baseD\baseGeq\baseC$, atomic multisets $M$ and atoms $q$, if, for all $\baseE\baseGeq\baseD$, $\suppM{\baseE}{\emptymultiset}\varphi$ implies $\suppM{\baseE}{M}q$, then $\suppM{\baseD}{L\msetsum M}q$. We note that in the case when $\baseE=\baseD=\baseC$,  $M=K$ and $q=p$ we have that $\suppM{\baseC}{\emptymultiset}\varphi$ implies $\suppM{\baseC}{M}p$ since the conclusion holds by hypothesis. Thus, we obtain $\suppM{\baseC}{L\msetsum M}p$, as required.
\end{proof}

We are now ready to prove the main results of this paper, that this semantics is indeed sound and complete for \ILL{}.

\section{Soundness}\label{sec:Soundness} 
\begin{theorem}[Soundness]\label{thm:Soundness}
    If $\Gamma \provesILL \varphi$ then $\Gamma \suppM{}{}\varphi$.
\end{theorem}
\begin{proof}
    By the inductive definition of $\provesILL$, it suffices to prove the following: 
    \begin{description}
        \item[Ax\label{eq:soundness-axiom}] $\varphi \suppM{}{} \varphi$ 
        \item[$\mto$I\label{eq:soundness-implication-intro}] If $\Gamma\msetsum \varphi \suppM{}{} \psi$ then $\Gamma \suppM{}{} \varphi \mto \psi$. 
        \item[$\mto$E\label{eq:soundness-implication-elim}] If $\Gamma \suppM{}{} \varphi \mto \psi$ and $\Delta \suppM{}{} \varphi$ then $\Gamma\msetsum \Delta \suppM{}{} \psi$. 
        \item[$\mand$I\label{eq:soundness-mult-conjunction-intro}] If $\Gamma \suppM{}{} \varphi$ and $\Delta \suppM{}{} \psi$ then $\Gamma\msetsum \Delta \suppM{}{} \varphi \mand \psi$. 
        \item[$\mand$E\label{eq:soundness-mult-conjunction-elim}] If $\Gamma \suppM{}{} \varphi \mand \psi$ and $\Delta\msetsum \varphi\msetsum \psi \suppM{}{} \chi$ then $\Gamma\msetsum \Delta \suppM{}{} \chi$. 
        \item[$\mtop$I\label{eq:soundness-mtop-intro}] $\suppM{}{} \mtop$
        \item[$\mtop$E\label{eq:soundness-mtop-elim}] If $\Gamma \suppM{}{} \mtop$ and $\Delta \suppM{}{} \varphi$ then $\Gamma\msetsum \Delta \suppM{}{} \varphi$.
        \item[$\aand$I\label{eq:soundness-add-conjunction-intro}] If $\Gamma \suppM{}{} \varphi$ and $\Gamma \suppM{}{} \psi$ then $\Gamma \suppM{}{} \varphi \aand \psi$.
        \item[$\aand$E\label{eq:soundness-add-conjunction-elim}] If $\Gamma \suppM{}{} \varphi \aand \psi$ then $\Gamma \suppM{}{} \varphi$ and $\Gamma \suppM{}{} \psi$.
        \item[$\aor$I\label{eq:soundness-add-disjunction-intro}] If $\Gamma \suppM{}{} \varphi$ or $\Gamma \suppM{}{} \psi$ then $\Gamma \suppM{}{} \varphi \aor \psi$.
        \item[$\aor$E\label{eq:soundness-add-disjunction-elim}] If $\Gamma \suppM{}{} \varphi \aor \psi$ and $\Delta\msetsum \varphi \suppM{}{} \chi$ and $\Delta\msetsum \psi \suppM{}{} \chi$ then $\Gamma\msetsum \Delta \suppM{}{} \chi$.
        \item[$\top$I\label{eq:soundness-atop-intro}] $\Gamma\suppM{}{}\top$, for any $\Gamma$.
        \item[$\abot$E\label{eq:soundness-abot-elim}] If $\Delta \suppM{}{} \abot$ then $\Gamma\msetsum\Delta \suppM{}{} \chi$, for any $\Gamma$. 
        \item[Promotion\label{eq:soundness-bang-promotion}] If $\Gamma_1 \suppM{}{} \bang\psi_1,\dots,\Gamma_n \suppM{}{} \bang\psi_n$ and $\bang\psi_1\msetsum \dots\msetsum \bang\psi_n\suppM{}{} \varphi$ then $\Gamma_1\msetsum \dots\msetsum \Gamma_n \suppM{}{} \bang \varphi$.
        \item[Dereliction\label{eq:soundness-bang-dereliction}] If $\Gamma \suppM{}{} \bang\varphi$, and $\Delta,\varphi \suppM{}{} \psi$ then $\Gamma\msetsum \Delta \suppM{}{} \psi$ holds.
        \item[Weakening\label{eq:soundness-bang-weakening}] If $\Gamma \suppM{}{} \bang\varphi$ and $\Delta \suppM{}{} \psi$ then $\Gamma\msetsum \Delta \suppM{}{} \psi$.
        \item[Contraction\label{eq:soundness-bang-contraction}] If $\Gamma \suppM{}{} \bang\varphi$ and $\Delta\msetsum \bang\varphi\msetsum \bang\varphi \suppM{}{} \psi$ then $\Gamma\msetsum \Delta \suppM{}{} \psi$.
    \end{description}
    We now proceed through the cases, noting that~\eqref{eq:soundness-bang-promotion},~\eqref{eq:soundness-atop-intro} and~\eqref{eq:soundness-abot-elim} hold for all $n\geq 0$.

    \begin{itemize}[label={-}]
        \item~\eqref{eq:soundness-axiom} This case is immediate.
        \item~\eqref{eq:soundness-implication-intro} We suppose $\Gamma\msetsum\varphi\suppM{}{}\psi$ and want to show $\Gamma\suppM{}{}\varphi\mto\psi$. To this end, it suffices to show that for all $\baseB$ and atomic multisets $L$ such that $\suppM{\baseB}{L}\Gamma$ and we have that $\varphi\suppM{\baseB}{L}\psi$ implies $\suppM{\baseB}{L}\varphi\mto\psi$. This follows immediately by~\ref{BeS:ILL:mto}.
        \item~\eqref{eq:soundness-implication-elim} We suppose that $\Gamma \suppM{}{} \varphi \mto \psi$ and $\Delta \suppM{}{} \varphi$ and want to show that $\Gamma\msetsum \Delta \suppM{}{} \psi$. It suffices to show that if for some $\baseB$ and atomic multisets $L$ and $K$ such that $\suppM{\baseB}{L}\Gamma$ and $\suppM{\baseB}{K}\Delta$ then $\suppM{\baseB}{L}\varphi\mto\psi$ and $\suppM{\baseB}{K}\varphi$ imply $\suppM{\baseB}{L\msetsum K}\psi$. To this end, we know that $\suppM{\baseB}{L}\varphi\mto\psi$ is equivalent to $\varphi\suppM{\baseB}{L}\psi$ by~\ref{BeS:ILL:mto} so we expand it by~\ref{BeS:ILL:inf} to get that for all $\baseC\baseGeq\baseB$ and atomic multisets $M$ such that $\suppM{\baseC}{M}\varphi$ implies $\suppM{\baseC}{L\msetsum M}\psi$. Since we have by hypothesis that $\suppM{\baseB}{K}\varphi$, we consider this implication under the assignments $\baseC=\baseB$ and $M=K$ to conclude $\suppM{\baseB}{L\msetsum K}\psi$, as required.
        \item~\eqref{eq:soundness-mult-conjunction-intro} We suppose that $\Gamma \suppM{}{} \varphi$ and $\Delta \suppM{}{} \psi$ and want to show that $\Gamma\msetsum \Delta \suppM{}{} \varphi \mand \psi$. It suffices to show that, if, for some $\baseB$ and atomic multisets $L$ and $K$ such that $\suppM{\baseB}{L}\Gamma$ and $\suppM{\baseB}{K}\Delta$, that $\suppM{\baseB}{L}\varphi$ and $\suppM{\baseB}{K}\psi$ imply $\suppM{\baseB}{L\msetsum K}\varphi\mand\psi$. To show $\suppM{\baseB}{L\msetsum K}\varphi\mand\psi$, by~\ref{BeS:ILL:mand}, we further suppose that for all $\baseC\baseGeq\baseB$, atomic multisets $M$ and atoms $p$, we have that $\varphi\msetsum\psi\suppM{\baseC}{M}p$. Our goal is to show $\suppM{\baseC}{L\msetsum K\msetsum M}p$. We note that by Lemma~\ref{lem:monotone-support} we have that $\suppM{\baseC}{L}\varphi$ and $\suppM{\baseC}{K}\psi$. By~\ref{BeS:ILL:inf}, we have that $\varphi\msetsum\psi\suppM{\baseC}{M}p$ is equivalent to considering all bases $\baseD\baseGeq\baseC$ and atomic multisets $N$ and $P$ such that $\suppM{\baseD}{N}\varphi$ and $\suppM{\baseD}{P}\psi$ imply $\suppM{\baseD}{N\msetsum P\msetsum M}p$. Since we have that $\suppM{\baseC}{L}\varphi$ and $\suppM{\baseC}{K}\psi$, then we consider this implication under the assignments $\baseD=\baseC$, $N=L$ and $P=K$ to get that $\suppM{\baseC}{L\msetsum K\msetsum M}p$, as required.
        \item~\eqref{eq:soundness-mult-conjunction-elim} We suppose that $\Gamma \suppM{}{} \varphi \mand \psi$ and $\Delta\msetsum \varphi\msetsum \psi \suppM{}{} \chi$ and want to show that $\Gamma\msetsum \Delta \suppM{}{} \chi$. It suffices to show that, if, for some $\baseB$ and atomic multisets $L$ and $K$ such that $\suppM{\baseB}{L}\Gamma$ and $\suppM{\baseB}{K}\Delta$, that $\suppM{\baseB}{L}\varphi\mand\psi$ and $\varphi\msetsum\psi\suppM{\baseB}{K}\chi$ imply $\suppM{\baseB}{L\msetsum K}\chi$. This follows immediately by Lemma~\ref{lem:mand-key-lemma}.
        \item~\eqref{eq:soundness-mtop-intro} We want to show that $\suppM{}{}\mtop$. This follows immediately by~\ref{BeS:ILL:mtop}.
        \item~\eqref{eq:soundness-mtop-elim} We suppose that $\Gamma \suppM{}{} \mtop$ and $\Delta \suppM{}{} \varphi$ and want to show that $\Gamma\msetsum \Delta \suppM{}{} \varphi$. It suffices to show that, if, for some $\baseB$ and atomic multisets $L$ and $K$ such that $\suppM{\baseB}{L}\Gamma$ and $\suppM{\baseB}{K}\Delta$, that $\suppM{\baseB}{L}\mtop$ and $\suppM{\baseB}{K}\varphi$ imply $\suppM{\baseB}{L\msetsum K}\varphi$. This follows immediately by Lemma~\ref{lem:mtop-key-lemma}.
        \item~\eqref{eq:soundness-add-conjunction-intro} We assume $\Gamma \suppM{}{} \varphi$ and $\Gamma \suppM{}{} \psi$. Fix $\baseB$ and $L$ such that $\suppM{\baseB}{L} \Gamma$. Thus, by (Inf), we have that $\suppM{\baseB}{L} \varphi$ and $\suppM{\baseB}{L} \psi$. Thus by ($\aand$) we have $\suppM{\baseB}{L} \varphi \aand \psi$, and thus by (Inf) we conclude $\Gamma \suppM{}{} \varphi \aand \psi$.
        \item~\eqref{eq:soundness-add-conjunction-elim} We assume $\Gamma \suppM{}{} \varphi \aand \psi$. Fix $\baseB$ and $L$ such that $\suppM{\baseB}{L}\Gamma$, we then have by (Inf) that $\suppM{\baseB}{L}\varphi \aand \psi$. By ($\aand$) we thus get that $\suppM{\baseB}{L} \varphi$ and $\suppM{\baseB}{L} \psi$, which by (Inf) gives $\Gamma \suppM{}{} \varphi$ and $\Gamma \suppM{}{} \psi$, as required. 
        \item~\eqref{eq:soundness-add-disjunction-intro} We assume $\Gamma \suppM{}{} \varphi$ or $\Gamma \suppM{}{} \psi$ holds. Fix $\baseB$ and $L$ such that $\suppM{\baseB}{L} \Gamma$ holds. Thus, by (Inf) we have that $\suppM{\baseB}{L} \varphi$ or $\suppM{\baseB}{L} \psi$ hold. Thus by ($\aor$), we have that $\suppM{\baseB}{L} \varphi \aor \psi$ holds. Thus by (Inf) we conclude $\Gamma \suppM{}{} \varphi \aor \psi$, as required.
        \item~\eqref{eq:soundness-add-disjunction-elim} We suppose $\Gamma \suppM{}{} \varphi \aor \psi$ and that both $\Delta, \varphi \suppM{}{} \chi$ and $\Delta, \psi \suppM{}{} \chi$ hold and want to show that $\Gamma, \Delta \suppM{}{}\chi$. By (Inf), it suffices to show that given:\\ 
        \begin{enumerate*}
            \item $\suppM{\baseB}{L}\varphi\aor\psi$\hspace{5pt}
            \item $\varphi\suppM{\baseB}{K}\chi$\hspace{5pt}
            \item $\psi\suppM{\baseB}{K}\chi$\hspace{5pt}
        \end{enumerate*}
        then $\suppM{\baseB}{L\msetsum K}\chi$ holds. This is immediate by Lemma~\ref{lem:aor-key-lemma}.
        \item~\eqref{eq:soundness-atop-intro} The conclusion follows immediately by ($\top$).
        \item~\eqref{eq:soundness-abot-elim} Start by fixing an arbitrary base $\baseB$ and atomic multisets $L$ and $K$ and multiset $\Gamma$ such that $\suppM{\baseB}{K}\Gamma$ and $\suppM{\baseB}{L}\Delta$. Thus it follows that $\suppM{\baseB}{L}\abot$. It now suffices to show that $\suppM{\baseB}{L\msetsum K}\chi$. This follows immediately by induction over the structure of $\chi$.
        \item~\eqref{eq:soundness-bang-promotion} We start by fixing an arbitrary $n \geq 0$. Then, by Corollary~\ref{cor:bang-promotion}, we have that the second hypothesis gives that $\bang\psi_1\msetsum \dots\msetsum \bang\psi_n\suppM{}{} \bang\varphi$. Now, by (Inf), if we consider all bases $\baseB$ and multisets $K_i$ such that $\suppM{\baseB}{K_i}\Gamma_i$ and say $K=K_1\msetsum\dots\msetsum K_n$, then it follows that $\suppM{\baseB}{K_i}\bang \psi_i$. Thus, by Corollary~\ref{cor:bang-cut-support}, we obtain that $\suppM{\baseB}{K}\bang\varphi$, which by (Inf) gives our desired result $\Gamma_1\msetsum\dots\msetsum\Gamma_n\suppM{}{}\bang\varphi$. 
        \item~\eqref{eq:soundness-bang-dereliction} It suffices to show that given $\suppM{\baseB}{L}\bang \varphi$ and $\varphi\suppM{\baseB}{K}\psi$ that $\suppM{\baseB}{L\msetsum K}\psi$ holds. We know that the second hypothesis, by Lemma~\ref{lem:bang-dereliction}, implies that $\bang\varphi\suppM{\baseB}{K}\psi$. This, together with the hypothesis that $\suppM{\baseB}{L}\bang \varphi$, by Lemma~\ref{lem:bang-cut-support-prim}, gives $\suppM{\baseB}{L\msetsum K}\psi$, as required.
        \item~\eqref{eq:soundness-bang-weakening} It suffices to show, given $\suppM{\baseB}{L}\bang\varphi$ and $\suppM{\baseB}{K}\psi$ that $\suppM{\baseB}{L\msetsum K}\psi$. By Lemma~\ref{lem:bang-mtop-interplay}, we have that the first hypothesis implies $\suppM{\baseB}{L}\mtop$. Similarly, the second hypothesis implies $\mtop\suppM{\baseB}{K}\psi$. Thus, we conclude that $\suppM{\baseB}{L\msetsum K}\psi$, as required.
        \item~\eqref{eq:soundness-bang-contraction} It suffices to show that given $\suppM{\baseB}{L}\bang\varphi$ and $\bang\varphi\msetsum\bang\varphi\suppM{\baseB}{K}\psi$ we can obtain $\suppM{\baseB}{L\msetsum K}\psi$. By (Inf), we observe that the second hypothesis is equivalent to $\bang\varphi\suppM{\baseB}{K}\psi$. From here, by Corollary~\ref{cor:bang-cut-support}, we obtain $\suppM{\baseB}{L\msetsum K}\psi$, as required.
    \end{itemize}
    This completes the proof of all items. 
\end{proof}

\section{Completeness}\label{sec:Completeness}
We now show that, given an arbitrary valid sequent $(\Gamma : \varphi)$ in our semantics, there exists a valid $\rm N_{ILL}$ proof of it. To this end, we will construct a special base, called $\baseILL$, whose rules will mimic the rules of the natural deduction system $\rm N_{ILL}$, with basic sentences ``simulating'' the subformulae of the arbitrary valid sequent. The key step will be to then show that derivations in $\baseILL$ directly correspond to natural deduction derivations of the formulae being simulated. Thus, we show that $(\Gamma : \varphi)$ is provable in $\rm N_{ILL}$ by, effectively, constructing the proof.

Let us begin by fixing an arbitrary sequent $\mathfrak{S}=(\Gamma:\varphi)$ and let $\Xi$ be the set of subformulae of the sequent $\mathfrak{S}$. That is to say, $\Xi$ is the union of the subformulae of each element of $\Gamma$ and $\varphi$. We additionally fix an injection $\flatILL{(\cdot)}:\Xi\rightarrow\At$, called the flattening map, such that:
\begin{itemize}
    \item It is the identity map on atoms and the units $\top$, $\abot$ and $\mtop$.
    \item For non-atomic formulae, $\phi$, it assigns an atom $p$ where $p\notin\Xi$ and for all $\alpha,\beta\in \Xi$, if $\alpha\neq\beta$ then $\flatILL{\alpha}\neq\flatILL{\beta}$.
\end{itemize}
Such a map has a left inverse, $\deflatILL{(\cdot)}$ defined similarly as:
\begin{itemize}
    \item The identity map on the  units $\top$, $\abot$ and $\mtop$ and on atoms not in the image of $\flatILL{(\cdot)}$.
    \item The original formula, i.e. $\deflatILL{(\flatILL{(\varphi)})} = \varphi$.
\end{itemize}
We further define these functions to be distributing over mulitsets; that is, given multisets $\Gamma=\{\gamma_1,\dots,\gamma_n\}$ and $P=\{p_1,\dots,p_n\}$ then $\flatILL{\Gamma} = \{\flatILL{\gamma_1},\dots,\flatILL{\gamma_n}\}$ and $\deflatILL{P} = \{\deflatILL{p_1},\dots,\deflatILL{p_n}\}$. We now define the simulation base $\baseILL$ relative to $\Xi$ and $\flatILL{(\cdot)}$ according to the rules of Figure~\ref{fig:SimulationBase} where $\varphi$ and $\psi$ range over all elements of $\Xi$ and $p$ ranges over all atoms $\At$. 

\begin{figure}[t]\captionsetup{justification=centering,font=small}
    \hrule \vspace{2mm}
    \small\[{
    \begin{array}{c@{\quad}c} 
    \infer[\irn \mto]{\phi\mto\psi}{\left\openaddrule\raisebox{-0.8em}{\deduce{\flatILL{\psi}}{[\flatILL{\phi}]}}\right\closeaddrule} & \infer[\ern \mto]{\flatILL{\psi}}{\openaddrule\flatILL{(\phi\mto\psi)}\closeaddrule & \openaddrule\flatILL{\phi}\closeaddrule} \\[3mm]
    \infer[\irn \mand]{\flatILL{(\phi\mand\psi)}}{\openaddrule\flatILL{\phi}\closeaddrule & \openaddrule\flatILL{\psi}\closeaddrule} & \infer[\ern \mand]{p}{\openaddrule\flatILL{\phi\mand\psi}\closeaddrule & \left\openaddrule\raisebox{-0.5em}{\deduce{p}{[\flatILL{\phi\msetsum\psi}]}}\right\closeaddrule}\\[3mm]
    \infer[\irn \mtop]{\flatILL{\mtop}}{} & \infer[\ern \mtop]{\flatILL{\phi}}{\openaddrule\flatILL{\mtop}\closeaddrule & \openaddrule\flatILL{\phi}\closeaddrule}\\[3mm]
    \infer[\irn \aand]{\flatILL{(\phi\aand\psi)}}{\openaddrule \flatILL{\phi} & \flatILL{\psi} \closeaddrule} & \infer[\ern{\aand_1}]{\flatILL{\phi}}{\openaddrule\flatILL{(\phi\aand\psi)}\closeaddrule} \quad \infer[\ern{\aand_2}]{\flatILL{\psi}}{\openaddrule\flatILL{(\phi\aand\psi)}\closeaddrule}\\[3mm]
    \infer[\irn{\aor_1}]{\flatILL{(\phi\aor\psi)}}{\openaddrule\flatILL{\phi}\closeaddrule} \quad \infer[\irn{\aor_2}]{\flatILL{(\phi\aor\psi)}}{\openaddrule\flatILL{\psi}\closeaddrule}  & \infer[\ern \aor]{p}{\openaddrule\flatILL{\phi\aor\psi}\closeaddrule & \left\{\raisebox{-0.5em}{ \deduce{p}{[\flatILL{\phi}]} \quad \deduce{p}{[\flatILL{\psi}]}} \right\}}\\[3mm]
    \infer[\irn \top]{\flatILL{\top}}{\openaddrule\emptymultiset\closeaddrule} & \infer[\ern \abot]{p}{\openaddrule\emptymultiset\closeaddrule & \openaddrule\flatILL{\abot}\closeaddrule}\\[3mm]
    \infer[\promotion]{\flatILL{\bang\phi}}{\llbracket\flatILL{\varphi}\rrbracket} & \infer[\dereliction]{\flatILL{\psi}}{\openaddrule\flatILL{(\bang\phi)}\closeaddrule & \left\openaddrule\raisebox{-0.8em}{\deduce{\flatILL{\psi}}{[\flatILL{\phi}]}}\right\closeaddrule}\\[3mm]
    \infer[\weakening]{\flatILL{\psi}}{\openaddrule\flatILL{(\bang\phi)}\closeaddrule & \openaddrule\flatILL{\psi}\closeaddrule} & \infer[\contraction]{\flatILL{\psi}}{\openaddrule\flatILL{(\bang\phi)}\closeaddrule & \left\openaddrule\raisebox{-0.8em}{\deduce{\flatILL{\psi}}{[\flatILL{(\bang\phi)}\msetsum\flatILL{(\bang\phi)}]}}\right\closeaddrule}\\[3mm]
    \end{array}}
    \] 
    \vspace{-1mm}
    
    \caption{The simulation base $\baseILL$.} \vspace{2mm}
    \hrule
    \label{fig:SimulationBase}
\end{figure}
As a result of our definition of $\baseILL$, we note that the only persistent atoms of $\baseILL$ are those atoms $\flatILL{(\bang\varphi)}$ where $\bang\varphi\in\Xi$. We now begin by proving completeness, making use of lemmas that we will prove later in this section.
\begin{theorem}[Completeness]\label{thm:completeness}
    If $\Gamma\suppM{}{}\varphi$ then $\Gamma\provesILL\varphi$.
\end{theorem}
\begin{proof}
    Let $\flatILL{(\cdot)}$ be a flattening map with $\deflatILL{(\cdot)}$ its corresponding inverse and $\baseILL$ be the simulation base for the sequent $(\Gamma:\varphi)$ as defined above. Since, by hypothesis, $\Gamma\suppM{}{}\varphi$, then in particular, it holds that $\Gamma\suppM{\baseILL}{\emptymultiset}\varphi$. By~\ref{BeS:ILL:inf}, our hypothesis is equivalent to considering an arbitrary base $\baseB\baseGeq\baseILL$ and atomic multiset $L$, such that $\suppM{\baseB}{L}\Gamma$ implies $\suppM{\baseB}{L}\varphi$. By Lemma~\ref{lem:completeness-atomic-simulation}, this is equivalent to considering an arbitrary base $\baseB\baseGeq\baseILL$ and atomic multiset $L$ where $\suppM{\baseB}{L}\flatILL{\Gamma}$ implies $\suppM{\baseB}{L}\flatILL{\varphi}$. Thus, we have, by~\ref{BeS:ILL:inf}, that $\flatILL{\Gamma}\suppM{\baseILL}{\emptymultiset}\flatILL{\varphi}$. By~\ref{BeS:ILL:at} and Lemma~\ref{lem:atomic-sound-and-complete}, it thus follows that $\flatILL{\Gamma}\deriveBaseM{\baseILL}\flatILL{\varphi}$, wherefore, by Lemma~\ref{lem:completeness-atomic-deflat}, $\Gamma\provesILL\varphi$ follows, as required.
\end{proof}

\begin{lemma}~\label{lem:completeness-atomic-simulation}
    For any $\baseB\baseGeq\baseILL$, $\suppM{\baseB}{L}\varphi$ if and only if $\suppM{\baseB}{L}\flatILL{\varphi}$.
\end{lemma}

\begin{proof}
    We only consider the case of $\varphi=\bang\alpha$ for some $\alpha\in\Xi$, as the rest follow suit.
    We proceed by induction on the structure of $\varphi$. By the definition of $(\bang)$ we have that $\suppM{\baseB}{L}\bang\varphi$ if and only if, for all $\baseC\baseGeq\baseB$, atomic multisets $K$ and atoms $p$, if, for all $\baseD\baseGeq\baseC$ such that $\suppM{\baseD}{\emptymultiset}\varphi$ implies $\suppM{\baseD}{K}p$, then $\suppM{\baseC}{L\msetsum K}p$.
    By the inductive hypothesis, we therefore have that this is equivalent to considering all bases $\baseC\baseGeq\baseB$, atomic multisets $K$ and atoms $p$, if, for all $\baseD\baseGeq\baseC$ such that $\deriveBaseM{\baseD}\flatILL{\varphi}$ implies $K\deriveBaseM{\baseD}p$, then $L\msetsum K\deriveBaseM{\baseC}p$. This, by Lemma~\ref{lem:completeness-atomic-definitions}, is equivalent to $\deriveBaseM{\baseB}\flatILL{(\bang\varphi)}$, as required.
\end{proof}

\begin{lemma}~\label{lem:completeness-atomic-deflat}
    If $L\deriveBaseM{\baseILL}p$ then $\deflatILL{L}\provesILL\deflatILL{p}$.
\end{lemma}

\begin{proof}
    It is clear that if $L\deriveBaseM{\baseILL}p$ holds due to~\eqref{eq:derive-ref} then $L=p$ and so we immediately have that $\deflatILL{p}\provesILL \deflatILL{p}$, as required. Else, it is the case that $L\deriveBaseM{\baseILL}p$ holds due to~\eqref{eq:derive-app}. In this case, we know that there must exists a rule in $\baseILL$ which, when applied, allows us to conclude $L\deriveBaseM{\baseILL}p$. This rule must be a rule of $\baseILL$, which, when one considers that for each formula $\varphi$ mentioned in each rule, it is the case that $\deflatILL{(\flatILL{(\varphi)})}=\varphi$, we see that we quickly recover instances of the rules of Figure~\ref{fig:NatDedILL2}, that is, the natural deduction system $\rm N_{ILL}$. Thus, we must be careful that by~\eqref{eq:derive-app}, we are not able to derive anything more than is possible in $\rm N_{ILL}$. For all rules, this is immediate except, perhaps, for $\promotion^\flat$, as one might question which atoms are persistent. However, as mentioned, since the only persistent atoms in $\baseILL$ are those which simulate formulae of the form $\bang\varphi$ by the definition of $\baseILL$, then this case becomes immediate, though we show this explicitly below:
    If $L\deriveBaseM{\baseILL}p$ holds due to the $\promotion^\flat$ rule, then $p=\flatILL{(\bang\varphi)}$ and there is a partition of $L$ into $L_1\msetsum\dots\msetsum L_n$, where $n\geq 0$, and some multiset of persistent atoms $D = \{d_1,\dots,d_n\}$ such that $L_i\deriveBaseM{\baseILL}d_i$ for all $i\in[1,n]$ and $D\deriveBaseM{\baseILL}\flatILL{\varphi}$. Since the only persistent atoms in $\baseILL$ are of the form $\flatILL{(\bang\alpha)}$ for $\bang\alpha\in\Xi$, thus we have that $d_i=\flatILL{((\bang\alpha)_i)}$. Thus, by the inductive hypothesis, we have that $\deflatILL{L}_i\provesILL(\bang\alpha)_i$ for all $i\in[1,n]$ and that $(\bang\alpha)_1\msetsum\dots(\bang\alpha)_n\provesILL\varphi$ all hold. Thus, by the $\promotion$ rule, we obtain that $\deflatILL{L_1}\msetsum\dots\msetsum\deflatILL{L_n}\provesILL\bang\varphi$ which is nothing more than $\deflatILL{L}\provesILL\deflatILL{p}$, as required.
\end{proof}

\begin{lemma}~\label{lem:completeness-structural-cut}
Given an arbitrary atom $p$, atomic multiset $K$, base $\baseB\baseGeq\baseILL$ and $\varphi\in\Xi$, then the following statements are equivalent:
\begin{enumerate}
    \item $\flatILL{(\bang\varphi)}\msetsum K\deriveBaseM{\baseB}p$.~\label{eq:structural-cut-1}
    \item For all $\baseC\baseGeq\baseB$, if $\deriveBaseM{\baseC}\flatILL{\varphi}$ then $K\deriveBaseM{\baseC}p$.~\label{eq:structural-cut-2}
\end{enumerate}
\end{lemma}

\begin{proof}
    To show~\eqref{eq:structural-cut-1} implies~\eqref{eq:structural-cut-2}, we first assume an arbitrary $\baseC\baseGeq\baseB$ such that $\deriveBaseM{\baseC}\flatILL{\varphi}$. Since $\deriveBaseM{\baseC}\flatILL{\varphi}$ holds, we have, by applying the $\promotion^\flat$ rule using~\eqref{eq:derive-app}, that $\deriveBaseM{\baseC}\flatILL{(\bang\varphi)}$. Since, by monotonicity (Lemma~\ref{lem:monotone-derivability}), we have that $\flatILL{(\bang\varphi)}\msetsum K\deriveBaseM{\baseC}p$, we can therefore use Lemma~\ref{lem:atomic-cut} to obtain $K\deriveBaseM{\baseC}p$, as required.\\
    To show~\eqref{eq:structural-cut-2} implies~\eqref{eq:structural-cut-1}, we first fix a base $\baseC=\baseB\,\cup\,\{\seq \flatILL{\varphi}\}$. Since it immediately follows by~\eqref{eq:derive-app} that $\deriveBaseM{\baseC}\flatILL{\varphi}$ then, by~\eqref{eq:structural-cut-2}, we have that $K\deriveBaseM{\baseC}p$. We now do a case analysis on how $K\deriveBaseM{\baseC}p$ holds.
    \begin{itemize}
        \item If $K\deriveBaseM{\baseC}p$ holds by~\eqref{eq:derive-ref} then $K=p$ and thus $p\deriveBaseM{\baseB}p$. Since $\flatILL{(\bang\varphi)}\deriveBaseM{\baseB}\flatILL{(\bang\varphi)}$ also holds by~\eqref{eq:derive-ref}, then we can use~\eqref{eq:derive-app}, applying the rule $\weakening^\flat$, to conclude $\flatILL{(\bang\varphi)}\msetsum p\deriveBaseM{\baseB}p$, as required.
        \item Else $K\deriveBaseM{\baseC}p$ holds by~\eqref{eq:derive-app}. We break this into two subcases.
        \begin{itemize}
            \item In the event the rule $\mathcal{R}$ is $\seq \flatILL{\varphi}$, then we have that $K=\emptymultiset$ and $p = \flatILL{\varphi}$. By the inductive hypothesis, we therefore have that $\flatILL{(\bang\varphi)}\deriveBaseM{\baseB}\flatILL{\varphi}$, which holds by $\dereliction^\flat$.
            \item Else, there exists a rule $\mathcal{R}=\langle\mathbf{A},\mathbf{S},p\rangle\in \baseB$ where $|\mathbf{A}| = m$, a partition of $K$ into $K_1\msetsum\dots\msetsum K_n$ and a multiset of persistent atoms $D=\{d_{m+1},\dots,d_{n}\}$ such that for all $\mathbf{T}_i\in\mathbf{A}$ and $Q\seq q\in \mathbf{T}_i$ we have that $K_i\msetsum Q\deriveBaseM{\baseC}q$, for all $i\in [1,m]$, and $K_{m+i}\deriveBaseM{\baseC}d_{m+i}$, for all $i\in [1,n-m]$, and for all $U\seq v \in \mathbf{S}$ we have that $D\msetsum U\deriveBaseM{\baseC}v$. By the inductive hypothesis, we therefore have that for all $\mathbf{T}_i\in\mathbf{A}$ and $Q\seq q\in \mathbf{T}_i$ that $\flatILL{(\bang\varphi)}\msetsum K_i\msetsum Q\deriveBaseM{\baseB}q$, for all $i\in[1,m]$, that $\flatILL{(\bang\varphi)}\msetsum K_{m+i}\deriveBaseM{\baseB}d_{m+i}$, for all $i\in [1, n-m]$ and that for all $U\seq v \in \mathbf{S}$ that $\flatILL{(\bang\varphi)}\msetsum D\deriveBaseM{\baseB}v$. Since $\flatILL{(\bang\varphi)}$ is a persistent atom in $\baseILL$ (due to the fact that $\promotion\in\baseILL$) and $\baseB\baseGeq\baseILL$, we therefore have that $\flatILL{(\bang\varphi)}$ is a persistent atom in $\baseB$. Furthermore, since $\flatILL{(\bang\varphi)}\deriveBaseM{\baseB}\flatILL{(\bang\varphi)}$ by~\eqref{eq:derive-ref}, we can therefore apply the rule $\mathcal{R}$ by~\eqref{eq:derive-app} to obtain $\flatILL{(\bang\varphi)}\msetsum K_1\msetsum\dots\msetsum\flatILL{(\bang\varphi)}\msetsum K_n\msetsum\flatILL{(\bang\varphi)}\deriveBaseM{\baseB}p$, which we can rewrite as $(\flatILL{(\bang\varphi)})^{n+1}\msetsum K\deriveBaseM{\baseB}p$. By repeatedly applying the $\contraction$ rule, we can reduce $(\flatILL{(\bang\varphi)})^{n+1}\msetsum K\deriveBaseM{\baseB}p$ down to the required result, which is $\flatILL{(\bang\varphi)}\msetsum K\deriveBaseM{\baseB}p$.
        \end{itemize}
    \end{itemize}
\end{proof}

It is important to note that the final rule application works because the multiset $\flatILL{(\bang\varphi)}\msetsum D$ is still a multiset of persistent atoms. To apply the rule $\mathcal{R}$, we need, amongst other details, to have a derivation of every element of this multiset as the union of the hypotheses of each derivation forms the context of the conclusion of the rule. Indeed we have that $\flatILL{(\bang\varphi)}\msetsum K_{m+i}\deriveBaseM{\baseB}d_{m+i}$ for all $i\in[1,n-m]$ but also that by~\eqref{eq:derive-ref} that $\flatILL{(\bang\varphi)}\deriveBaseM{\baseB}\flatILL{(\bang\varphi)}$. Thus, the rule application results in $(\flatILL{(\bang\varphi)})^{n+1}\msetsum K\deriveBaseM{\baseB}p$ where the extra $\flatILL{(\bang\varphi)}$ arises as a result of this extra element of the multiset of persistent atoms.

\begin{lemma}~\label{lem:completeness-atomic-definitions}
    The following hold for all $\baseB\baseGeq\baseILL$ and atomic multisets $L$:
    \begin{enumerate}
        \item $L\deriveBaseM{\baseB}\flatILL{(\varphi\mand\psi)}$ iff for all $\baseC\baseGeq\baseB$, atomic multisets $K$ and atoms $p$, if $\flatILL{\varphi}\msetsum\flatILL{\psi}\msetsum K\deriveBaseM{\baseC}p$ then $L\msetsum K\deriveBaseM{\baseC}p$
        \item $L\deriveBaseM{\baseB}\flatILL{(\varphi\mto\psi)}$ iff $L\msetsum\flatILL{\varphi}\deriveBaseM{\baseB}\psi$
        \item $L\deriveBaseM{\baseB}\flatILL{\mtop}$ iff for all $\baseC\baseGeq\baseB$, atomic multisets $K$ and atoms $p$, if $K\deriveBaseM{\baseC}p$ then $L\msetsum K\deriveBaseM{\baseC}p$
        \item $L\deriveBaseM{\baseB}\flatILL{(\varphi\aand\psi)}$ iff $L\deriveBaseM{\baseB}\flatILL{\varphi}$ and $L\deriveBaseM{\baseB}\flatILL{\psi}$
        \item $L\deriveBaseM{\baseB}\flatILL{(\varphi\aor\psi)}$ iff for all $\baseC\baseGeq\baseB$, atomic multisets $K$ and atoms $p$, if $K\msetsum\flatILL{\varphi}\deriveBaseM{\baseC}p$ and $K\msetsum\flatILL{\psi}\deriveBaseM{\baseC}p$ then $L\msetsum K\deriveBaseM{\baseC}p$
        \item $L\deriveBaseM{\baseB}\flatILL{\top}$ always
        \item $L\deriveBaseM{\baseB}\flatILL{\abot}$ iff $L\msetsum K\deriveBaseM{\baseB}p$, for all atomic multisets $K$ and atoms $p$
        \item $L\deriveBaseM{\baseB}\flatILL{(\bang\varphi)}$ iff for all base $\baseC\baseGeq\baseB$, atomic multisets $K$ and atoms $p$, if for all bases $\baseD\baseGeq\baseC$ it holds that $\deriveBaseM{\baseD}\flatILL{\varphi}$ implies $K\deriveBaseM{\baseD}p$, then $L\msetsum K\deriveBaseM{\baseC}p$
    \end{enumerate}
\end{lemma}

\begin{proof}
    Here we only include the proof of the last case. The rest can be found in Appendix~\ref{secA:Proofs}.
    To prove the final point, we start by noting that our statement can be simplified by Lemma~\ref{lem:completeness-structural-cut}. Thus, we can restate our lemma to say $L\deriveBaseM{\baseB}\flatILL{(\bang\varphi)}$ iff for all base $\baseC\baseGeq\baseB$, atomic multisets $K$ and atoms $p$, if $\flatILL{(\bang\varphi)}\msetsum K\deriveBaseM{\baseC}p$, then $L\msetsum K\deriveBaseM{\baseC}p$. We now show this instead.
    Going left to right, we note that by monotonicity (Lemma~\ref{lem:monotone-derivability}) we have that $L\deriveBaseM{\baseC}\flatILL{(\bang\varphi)}$. Thus, by Lemma~\ref{lem:atomic-cut}, we therefore have $L\msetsum K\deriveBaseM{\baseC}p$, as required.\\
    Going right to left, we start by considering $\baseC=\baseB$, $K=\emptymultiset$ and $p=\flatILL{(\bang\varphi)}$. Thus, our hypothesis becomes if $\flatILL{(\bang\varphi)}\deriveBaseM{\baseB}\flatILL{(\bang\varphi)}$, then $L\deriveBaseM{\baseB}\flatILL{(\bang\varphi)}$. Since $\flatILL{(\bang\varphi)}\deriveBaseM{\baseB}\flatILL{(\bang\varphi)}$ holds by~\eqref{eq:derive-ref}, we thus get $L\deriveBaseM{\baseB}\flatILL{(\bang\varphi)}$, as required.
\end{proof}

\section{Comments on the semantics}~\label{sec:Conc}

As mentioned in the introduction, the work herein uses the solid framework set up by Gheorghiu, Gu, and Pym in~\cite{AlexTaoDavid_PtS4IMLL} to help us capture substructurality. However, extending to the additives and including the modality required \emph{quite} some additional mathematical machinery. Whilst the modification to~\ref{BeS:ILL:inf} was previously discussed in Section~\ref{sec:BeS}, what is not perhaps clear is the origins of the clause for~\ref{BeS:ILL:bang}. A na\"{i}ve but not incorrect view of the clause is that we simply put things in the right places to obtain an inductive definition that we can ``cut" resources on. This argument can be further justified if one makes use of the identities $\bang\top \equiv \mtop$ and $\bang(\varphi\aand\psi)\equiv(\bang\varphi)\mand(\bang\psi)$, where $\equiv$ represents logical equivalence. Supposing we accept these equivalences and the argument provided in Section~\ref{sec:BeS} for the modified~\ref{BeS:ILL:inf} clause, we can derive the clause for ($\bang$) as follows:
\begin{enumerate}
    \item Start with the fact that $\suppM{\baseB}{L}\bang(\varphi\aand\psi)$ iff $\suppM{\baseB}{L}(\bang\varphi)\mand(\bang\psi)$.
    \item Let $\psi = \top$. Thus, we have that $\suppM{\baseB}{L}\bang(\varphi\aand\top)$ iff $\suppM{\baseB}{L}(\bang\varphi)\mand(\bang\top)$.
    \item Since $\varphi\aand\top \equiv \varphi$ and $\bang\top \equiv \mtop$, we therefore have that $\suppM{\baseB}{L}\bang(\varphi)$ iff $\suppM{\baseB}{L}(\bang\varphi)\mand\mtop$.
    \item By~\ref{BeS:ILL:mand}, the right hand side becomes $\suppM{\baseB}{L}\bang\varphi$ iff for all $\baseC\baseGeq\baseB$, atomic multisets $K$ and atoms $p$, if $\varphi\msetsum\mtop\suppM{\baseC}{K}p$, then $\suppM{\baseC}{L\msetsum K}p$.
    \item Since $\varphi\msetsum\mtop\suppM{\baseX}{M}\psi$ iff $\varphi\suppM{\baseX}{M}\psi$, our equivalence becomes $\suppM{\baseB}{L}\bang\varphi$ iff for all $\baseC\baseGeq\baseB$, atomic multisets $K$ and atoms $p$, if $\bang\varphi\suppM{\baseC}{K}p$, then $\suppM{\baseC}{L\msetsum K}p$.
    \item Finally, by~\ref{BeS:ILL:inf}, we have that $\suppM{\baseB}{L}\bang\varphi$ iff for all $\baseC\baseGeq\baseB$, atomic multisets $K$ and atoms $p$, if, for all $\baseD\baseGeq\baseC$ such that $\suppM{\baseD}{\emptymultiset}\bang\varphi$ implies $\suppM{\baseD}{K}p$, then $\suppM{\baseC}{L\msetsum K}p$.
\end{enumerate}

\noindent This line of reasoning allows one to correctly deduce a clause for~\ref{BeS:ILL:bang} with the right formulaic equivalences, such that the formulae on the right-hand side of the clause are of strictly lower weight than on the left-hand side, as discussed in Section~\ref{sec:BeS}. However, doing so seems to suggest that the meaning of~\ref{BeS:ILL:bang} comes from the~\ref{BeS:ILL:inf} clause, which is at odds with the purpose of the~\ref{BeS:ILL:bang} clause. In this case, we observe that Theorem~\ref{thm:cut-support} puts this issue to rest, as regardless of whether one uses~\ref{BeS:ILL:gen-inf} or~\ref{BeS:ILL:inf}, we can use the~\ref{BeS:ILL:bang} clause without modification. Nevertheless, one might note that the core of the definition of~\ref{BeS:ILL:bang} exactly meets the criterion for formulae with~\ref{BeS:ILL:bang} as a top-level connective on the left-hand side of a sequent in the~\ref{BeS:ILL:inf}; that is, the requirement for an object of the form $\suppM{\baseX}{\emptymultiset}\varphi$. So the question arises, should it not be possible to define~\ref{BeS:ILL:bang} as something along the lines of
\[
    \suppM{\baseB}{L}\bang\varphi \text{ iff } \suppM{\baseB}{\emptymultiset}\varphi \text{ and }L = ?
\]
where $L=?$ is to mean that the condition on $L$ is unclear. If one takes $L=\emptymultiset$ then such a clause fails to be both sound and complete as the requirement that $L$ be empty cannot be enforced \emph{prima facie}. However, the clause~\ref{BeS:ILL:bang} can be viewed as a reflection (in a sense related to the works of Halln\"{a}s and Schroeder-Heister~\cite{SchroederHeister_UniformPtSforLogicalConstants_1991, Schroeder-Heister2007-SCHGDR, HallnasPsH-ProofTheoreticApproachToLP} and the work of Gheorghiu, Gu and Pym~\cite{AlexTaoDavid_PtS4IMLL}) of the expression above with $L=\emptymultiset$. 

The connection between the~\ref{BeS:ILL:bang} clause and the condition $L=\emptymultiset$ is perhaps best understood in the context of the works of Wadler~\cite{Wadler_TasteOfLL_1993} and Pfenning et al.~\cite{PFENNING_DAVIES_2001, ChangChaudhuriPfenning-JudgementalAnalysisOfLinearLogic} and Dual Intuitionistic Linear Logic of Barber~\cite{Barber1996DualIL} where, following Andreolli~\cite{Andreoli_LPwithFocusingProofsInLL_1992}, we represent sequents 

\[\Gamma\seq\varphi \text{ as } \Theta\ctxt\Delta\seq \varphi\]

where the $\ctxt$ is meant to represent a separation of ``types" of hypotheses. Following Pfenning et al., elements of $\Theta$ obtain the interpretation that they are in some sense ``valid" assumptions, whereas elements of $\Delta$ are as before. As shown in~\cite{PFENNING_DAVIES_2001, ChangChaudhuriPfenning-JudgementalAnalysisOfLinearLogic}, the elements of $\Delta$ contain formulae which we can move into $\Theta$. However, when we do so, we prepend each such formula with a $(\bang)$. Similarly, if we move a formula from $\Delta$ to $\Theta$, it must first have a $(\bang)$ as top-level connection and we strip it away when moving it over. If we were to setup a semantic theory along judgements of this form, one would give a clause for~\ref{BeS:ILL:bang} along the lines of
\begin{align*}
    \suppL{\baseB}{G}{L}\bang\varphi &\mbox{ iff for all } \baseC\baseGeq\baseB, \mbox{ atomic multisets } K \mbox{ and atoms } p,\\ 
    &\mbox{ if } \varphi\ctxt\emptymultiset\suppL{\baseC}{G}{K}p, \mbox{ then } \suppL{\baseC}{G}{L,K}p    
\end{align*}

Under the interpretation of Pfenning et al. we see that the real meaning of $\suppM{\baseB}{L}\bang\varphi$ is given by what one can do with it if one assumes $\varphi$ to be \emph{valid}\footnote{One may indeed setup a semantics along this lines and obtain soundness and completeness results following the results presented in this paper. Is is perhaps now clear that the clause for~\ref{BeS:ILL:inf} presented in the semantics of this paper actually comes from this line of work, as $\bang$ formulae are treated completely separately.}. Returning to the semantics presented in this paper, we see that the interpretation of $\suppM{\baseB}{L}\bang\varphi$ expresses this assumed validity by saying that for all $\baseC\baseGeq\baseB$, atomic multisets $K$ and atoms $p$, if for all $\baseD\baseGeq\baseC\, \suppM{\baseD}{\emptymultiset}\varphi$ implies $\suppM{\baseD}{K}p$, then $\suppM{\baseC}{L\msetsum K}p$. The validity of $\varphi$ is perfectly captured by this clause with the expression that we consider all extensions $\baseD$ of the base $\baseC$ where $\varphi$ is a ``theorem" of the base, with modal-like behaviour. This isn't to say that $\varphi$ is necessarily a theorem of $\baseC$ but we consider all extensions where it is. If anything follows from such an assumption, then it is as if we have the assumption that $\bang\varphi$ is on the left-hand side of the support judgement and thus we indeed \emph{can} cut on it. Furthermore, as a result of the fact that $\varphi$ is a theorem, we have the requirement of before that the assumption of a $\bang\varphi$, that is, the assumption that $\varphi$ is a theorem, allows us to deduce as many copies of $\varphi$ as we need, i.e. such formulae are structural in nature. Similarly, a reverse reading shows that indeed we are necessitating if we are to conclude $\bang\varphi$, i.e. capturing the modal-like behaviour previously mentioned. Thus, it should be understood that this clause really is intrinsically capturing the definition of~\ref{BeS:ILL:bang}. Thus, the role of the persistent atoms, as defined in Section~\ref{sec:Basic-Derivability} should perhaps be clearer. Such atoms are those which \emph{may} be considered theorems of the base, but more importantly, they are the atoms whose behaviour is modal in nature. This distinction is important because we do not wish to consider all theorems and axioms of the base when working with persistent atoms, only that subset which has modal-like behaviour, which for us means, can be introduced using promotion-like reasoning. This is, of course, necessary for our completeness argument. Note, importantly, that the definition of a persistent atom does \emph{not} require that the atom is structural in any way. This is perhaps surprising as it says that as far as the semantics is concerned, modal behaviour is still restricted to what it has always been: promotion (or in other words, necessitation).

To conclude, I would like to discuss the prospects for the embedding of the semantics for IPL in our semantics for ILL. We know from~\cite{girard_LLSyntaxAndSemantics,Girard_LinearLogic_1987,Bierman_OnILL_1994} that there are many possible translations of formulae from IPL to \ILL{}, called Girard translations. We present a particular one below:
\begin{definition}
    The mapping $\iplill{\cdot}:Form_\text{IPL}\rightarrow Form_\text{ILL}$ can be defined as follows:
    \begin{enumerate}
        \item $p \mapsto p$, where $p$ is a propositional atom
        \item $\varphi\land\psi \mapsto \iplill{\varphi}\aand \iplill{\psi}$ 
        \item $\varphi \lor\psi \mapsto \bang\iplill{\varphi} \aor \bang\iplill{\psi}$ 
        \item $\varphi \supset \psi \mapsto \bang\iplill{\varphi} \mto \iplill{\psi}$
        \item $\bot \mapsto \iplill{\abot}$
    \end{enumerate}
\end{definition}
Girard, in~\cite{girard_LLSyntaxAndSemantics}, says that the crux of the translation is the following: $(\Gamma:\varphi)_{\text{IPL}}$ is intuitionistically provable if and only if $(\bang\iplill{\Gamma}:\iplill{\varphi})_{\text{ILL}}$ is linearly provable, a relevant proof of which can be found in~\cite{Bierman_OnILL_1994}. This translation ties in to our intuition that structurality in the hypothesis of a sequent is really properly represented by our treatment of ($\bang$) in (Inf). Since we have a sound and complete P-tS for IPL and now a sound and complete P-tS for \ILL{}, we therefore know that we can always map any \emph{valid} IPL sequent $\Gamma\suppIPL{\emptybase}\varphi$ (using $\suppIPL{\baseX}$ for the support relation of Sandqvist's semantics in~\cite{Tor2015}) to a valid sequent in \ILL{} of the form $\bang\iplill{\Gamma}\suppM{\emptybase}{\emptymultiset}\iplill{\varphi}$, and vice-versa. This result is interesting as it gives a way of analysing valid sequents of IPL in the framework we have setup for \ILL{} which is certainly not without its quirks (for example, consider how disjunction maps over!). However, a natural question to ask would be how may one generalise this mapping? What if we were given a formula and base in which inference of the formula is supported i.e. $\suppIPL{\baseB}\varphi$, and wanted to try and understand it in the linear setting, i.e. to find a multiset $L$ and a base $\iplill{\baseB}$ such that the support relation $\suppM{\iplill{\baseB}}{L}\iplill{\varphi}$ now holds? Whilst it is obvious that the formula is mappable directly, we are then stuck with how to obtain $L$ and $\iplill{\baseB}$, as rules and atomic derivability in the two semantics are \emph{quite} differently behaved. At present, it is not clear to me how this mapping should be done, though I do believe such a mapping between the semantics of Sandqvist and ours for \ILL{} is possible. However, I believe that instead, the correct approach to take if we are committed to this line of investigation, is to define a support relation for IPL that is in some sense much closer to ours for \ILL{}, whose treatment of formulae is closer to our own and whose atomic derivability relation mirrors ours in how it uses rules of the base, and whose base rules may also include the additional structure that ours do. Whilst it can easily be shown that one can have a sound and complete P-tS for IPL which keeps track of atoms in much the same way as ours for \ILL{} does, I have had no luck in finding a way of constructing such a mapping. I thus leave this problem open for further study.

\backmatter
\begin{appendices}
\section{Omitted proofs from Sections~\ref{sec:BeS} and~\ref{sec:Completeness}}\label{secA:Proofs}
This appendix contains some lemmas and proofs that were deemed too long to include in the main body of the paper. They are included to provide a complete account of the results presented in the main body of the paper, however the technical details of the proofs are not particularly enlightening. In each case, to aid the reader, we restate the lemma before its' proof.

\begin{lemma*}
    Given $\suppM{\baseB}{L} \varphi \mand \psi$ and $\varphi\msetsum\psi \suppM{\baseB}{K} \chi$ then $\suppM{\baseB}{L\msetsum K} \chi$ holds.
\end{lemma*}

\begin{proof}[(Proof of Lemma~\ref{lem:mand-key-lemma})]
    In this proof, we only consider the additive connectives. For the multiplicative connectives, I refer the reader to~\cite{AlexTaoDavid_PtS4IMLL}.
    We proceed by proving by induction on the structure of $\chi$. 
    \begin{itemize}[label={-}] 
        \item $\chi=\alpha \aand \beta$. By ($\aand$), we see that we need to show that $\suppM{\baseB}{L\msetsum K}\alpha$ and $\suppM{\baseB}{L\msetsum K}\beta$. 
        
        The second hypothesis states that $\varphi\msetsum\psi \suppM{\baseB}{K}\alpha \aand \beta$ which by ($\aand$) and (Inf) gives $\varphi\msetsum\psi \suppM{\baseB}{K}\alpha$ and $\varphi\msetsum\psi \suppM{\baseB}{K}\beta$. Then, we apply the inductive hypothesis which gives us that $\suppM{\baseB}{L\msetsum K}\alpha$ and $\suppM{\baseB}{L\msetsum K}\beta$ which by ($\aand$) gives $\suppM{\baseB}{L\msetsum K}\alpha \aand \beta$, as required.
        \item $\chi=\alpha \aor \beta$. Spelling out the conclusion of the lemma gives that that we want to show that for all $\baseC \baseGeq \baseB$ atomic multisets $M$ and atoms $p$ such that $\alpha \suppM{\baseC}{M}p$ and $\beta \suppM{\baseC}{M}p$ implies that $\suppM{\baseC}{L\msetsum K\msetsum M}p$. 
        To do that it suffces to show the following two things:
        \begin{itemize}
            \item $\suppM{\baseC}{L}\varphi \mand \psi$. This holds by monotonicity from the first hypothesis.
            \item $\varphi\msetsum\psi \suppM{\baseC}{K\msetsum M} p$. To show this we suppose that we have for all $\baseD \baseGeq \baseC$ and atomic multisets $N$ that $\suppM{\baseD}{N}\varphi \msetsum\psi$. Thus, since $\varphi\msetsum\psi \suppM{\baseB}{K} \alpha\aor\beta$ we get that $\suppM{\baseD}{K\msetsum N} \alpha \aor \beta$. Unfolding the definition of ($\aor$) gives that we have that for all $\base{E} \baseGeq \baseD$, atomic multisets $M$ and atoms $p$ such that $\alpha\suppM{\base{E}}{M}p$ and $\beta\suppM{\base{E}}{M}p$ implies $\suppM{\base{E}}{K\msetsum N\msetsum M}p$ and in particular when $\base{E}=\baseD$ that $\suppM{\baseD}{K\msetsum N\msetsum M}p$. Thus we conclude that $\varphi\msetsum\psi \suppM{\baseC}{K\msetsum M}p$, as required.
        \end{itemize}
        To finish the argument, we unfold $\suppM{\baseC}{L}\varphi \mand \psi$ according to ($\mand$) which gives that for all $\baseD \baseGeq \baseC$, atomic multisets $V$ and atoms $p$ such that $\varphi\msetsum\psi \suppM{\baseD}{V}p$ implies that $\suppM{\baseD}{L\msetsum V}p$. In particular this holds when $\baseD=\baseC$ and $V = K\msetsum M$. Thus we conclude that $\suppM{\baseC}{L\msetsum K\msetsum M}p$, as required.
        \item $\chi=\abot$. To show this we start by unpacking $\varphi\msetsum\psi \suppM{\baseB}{K} \chi$ which gives that we have that for all atomic $p$ and $M$ that $\varphi\msetsum\psi \suppM{\baseB}{K\msetsum M} p$. Further unpacking $\suppM{\baseB}{L} \varphi \mand \psi$ according to ($\mand$) gives that for all $\baseC \baseGeq \baseB$, atomic multisets $V$ and atoms $p$, if $\varphi\msetsum\psi \suppM{\baseC}{V}p$ then $\suppM{\baseC}{V}p$. Since by hypothesis we have for all $p$ and $M$ that $\varphi\msetsum\psi \suppM{\baseB}{K\msetsum M} p$, we can conclude that $\suppM{\baseB}{L\msetsum K\msetsum Q} p$ for all $p$ and all $Q$ which equivalently gives $\suppM{\baseB}{L\msetsum K} \abot$, as required.
        \item $\chi=\bang \alpha$. Unfolding the conclusion gives that supposing that for all bases $\baseC\baseGeq\baseB$, atomic multisets $M$ and atoms $p$, such that $\bang \alpha\suppM{\baseC}{M}p$ we want to show that $\suppM{\baseC}{L\msetsum K\msetsum M}p$.
        To this end, we prove the following:
        \begin{itemize}
            \item $\suppM{\baseC}{L}\varphi \mand \psi$. This holds by monotonicity from the first hypothesis.
            \item $\varphi\msetsum\psi\suppM{\baseC}{K\msetsum M}p$. To show this, we start from the second hypothesis which gives by (Inf), for all bases $\baseD\baseGeq\baseC$ and atomic multisets $N$ such that $\suppM{\baseD}{N}\varphi\msetsum\psi$, then $\suppM{\baseD}{K\msetsum N}\bang\alpha$. This, by ($\bang$), is equivalent to considering further all extensions $\baseE\baseGeq\baseD$, atomic multisets $Q$ and atoms $p\in\At$ such that if $\bang\alpha\suppM{\baseE}{Q}p$ then $\suppM{\baseE}{K\msetsum N\msetsum Q}p$. By our additional hypothesis and monotonicity, if we consider the case when $\baseE=\baseD$ and $Q=M$, we obtain therefore $\suppM{\baseD}{K\msetsum N\msetsum M}p$, which by (Inf) gives $\varphi\msetsum\psi\suppM{\baseC}{K\msetsum M}p$, as required.
        \end{itemize}
        To finish the argument, we note that the first point is equivalent to saying for all $\baseX\baseGeq\baseC$, atomic multisets $N$, and atoms $p\in\At$, if $\varphi\msetsum\psi\suppM{\baseX}{N}p$ then we obtain $\suppM{\baseX}{L\msetsum N}p$. By the second point, setting $\baseX=\baseC$ and $N=K\msetsum M$ we thus obtain $\suppM{\baseC}{L\msetsum K\msetsum M}p$, as required.
    \end{itemize}
\end{proof}

\begin{lemma*}
    Given $\suppM{\baseB}{L} \mtop$ and $\suppM{\baseB}{K} \chi$ then $\suppM{\baseB}{L\msetsum K} \chi$ holds.
\end{lemma*}

\begin{proof}[(Proof of Lemma~\ref{lem:mtop-key-lemma})]
    Again, in this proof, we only consider the additive connectives. For the multiplicative connectives, I once more refer the reader to~\cite{AlexTaoDavid_PtS4IMLL}.
    We proceed by proving by induction on the structure of $\chi$. 
    \begin{itemize}[label={-}]
        \item $\chi=\alpha \aand \beta$. 
        We starting from $\suppM{\baseB}{K}\alpha \aand \beta$ which by ($\aand$) gives $\suppM{\baseB}{K}\alpha$ and $\suppM{\baseB}{K}\beta$. We then apply the inductive hypothesis from which it follows that $\suppM{\baseB}{L\msetsum K}\alpha$ and $\suppM{\baseB}{L\msetsum K}\beta$, which by ($\aand$) gives $\suppM{\baseB}{L\msetsum K}\alpha \aand \beta$, as required.
        \item $\chi=\alpha \aor \beta$. Unfolding the conclusion gives that for all $\baseC \baseGeq \baseB$, atomic multisets $M$ and atoms $p$ if $\alpha \suppM{\baseC}{M} p$ and $\beta \suppM{\baseC}{M}$ then $\suppM{\baseC}{L\msetsum K\msetsum M}p$. Thus given such an $\alpha \suppM{\baseC}{M}p$ and $\beta \suppM{\baseC}{M}$ we want to show $\suppM{\baseC}{L\msetsum K\msetsum M}p$. To show this, we do the following:
        \begin{itemize}
            \item[-] $\suppM{\baseC}{L}\mtop$. This follows by monotonicity.
            \item[-] $\alpha \suppM{\baseC}{K\msetsum M} p$. Starting from $\alpha \suppM{\baseC}{M} p$, by (Inf) we have that for all $\baseD \baseGeq \baseC$ and atomic multisets $N$ such that $\suppM{\baseD}{N} \alpha$ implies $\suppM{\baseD}{K\msetsum M\msetsum N} p$. By the previous fact, we thus have $\suppM{\baseD}{K\msetsum M\msetsum N} p$ which by (Inf) gives $\alpha \suppM{\baseC}{K\msetsum M} p$, as required.
            \item[-] $\beta \suppM{\baseC}{K\msetsum M} p$. The proof of this case is identical to the previous case.
        \end{itemize}
        We thus have sufficient grounds to use the second hypothesis to conclude $\suppM{\baseC}{L\msetsum K\msetsum M}p$, as required.
        \item $\chi=\abot$. Unfolding the second hypothesis gives us that for all atoms $p$ and $M$ we have $\suppM{\baseB}{K\msetsum M} p$. Using the first hypothesis we get that this implies that for all $p$ and $M$ that $\suppM{\baseB}{L\msetsum K\msetsum M} p$. Thus we conclude $\suppM{\baseB}{L\msetsum K} \abot$, as required.
        \item $\chi=\bang \alpha$. Unfolding the conclusion gives that supposing that for all bases $\baseC\baseGeq\baseB$, atomic multisets $M$ and atoms $p$, such that $\bang \alpha\suppM{\baseC}{M}p$ we want to show that $\suppM{\baseC}{L\msetsum K\msetsum M}p$.
        To this end, we prove the following:
        \begin{itemize}
            \item $\suppM{\baseC}{L}\mtop$. This holds by monotonicity from the first hypothesis.
            \item $\suppM{\baseC}{K\msetsum M}p$. To show this, we start from the second hypothesis which by ($\bang$) is equivalent to considering further all extensions $\baseC\baseGeq\baseB$, atomic multisets $N$ and atoms $p\in\At$ such that if $\bang\alpha\suppM{\baseC}{N}p$ then $\suppM{\baseC}{K\msetsum N}p$. By our additional hypothesis, we consider the case when $N=M$, thus giving $\suppM{\baseC}{K\msetsum M}p$, as required.
        \end{itemize}
        To finish the argument, we note that the first point is equivalent to saying for all $\baseX\baseGeq\baseC$, atomic multisets $N$, and atoms $p\in\At$, if $\suppM{\baseX}{N}p$ then we obtain $\suppM{\baseX}{L\msetsum N}p$. By the second point, setting $\baseX=\baseC$ and $N=K\msetsum M$ we thus obtain $\suppM{\baseC}{L\msetsum K\msetsum M}p$, as required.
    \end{itemize}
\end{proof}

\begin{lemma*}
     Given that $\suppM{\baseB}{L} \varphi \aor \psi$, $\varphi \suppM{\baseB}{K} \chi$ and $\psi \suppM{\baseB}{K} \chi$ all hold then $\suppM{\baseB}{L\msetsum K} \chi$.
\end{lemma*}
\begin{proof}[(Proof of Lemma~\ref{lem:aor-key-lemma})]
    In this case we consider the base case, one multiplicative and one additive case. The other cases follow similarly.
    We proceed by induction on the structure of $\chi$.    
    \begin{itemize}[label={-}]
        \item $\chi = p$, for atomic $p$. The second and third hypotheses combined give sufficient conditions to conclude from the first hypothesis and ($\aor$) that $\suppM{\baseB}{L\msetsum K} p$.this
        \item $\chi=\alpha \aand \beta$. From the second hypothesis and by ($\aand$) and (Inf) we get $\varphi \suppM{\baseB}{K}\alpha$ and $\varphi \suppM{\baseB}{K}\beta$. Arguing similarly for the third hypothesis we get $\psi \suppM{\baseB}{K}\alpha$ and $\psi \suppM{\baseB}{K}\beta$. Then by applying the inductive hypothesis we obtain that $\suppM{\baseB}{L\msetsum K}\alpha$ and $\suppM{\baseB}{L\msetsum K}\beta$. Thus, by ($\aand$), we conclude $\suppM{\baseB}{L\msetsum K}\alpha \aand \beta$.
        \item $\chi=\alpha \mand \beta$. Unfolding the conclusion gives that we want to show that for all $\baseC \baseGeq \baseB$, atomic mulitsets $M$, atoms $p$ such that $\alpha\msetsum\beta \suppM{\baseC}{M} p$ then we can conclude $\suppM{\baseC}{L\msetsum K\msetsum M} p$. To do this, we need to show three things:
        \begin{itemize}
            \item $\suppM{\baseC}{L} \varphi \aor \psi$. This follows by monotonicity.
            \item $\varphi \suppM{\baseC}{K\msetsum M} p$. To show this, suppose we have that for all $\baseD \baseGeq \baseC$ and atomic multisets $N$ such that $\suppM{\baseD}{N}\varphi$. Then we have by the second hypothesis that $\suppM{\baseD}{K\msetsum N} \alpha \mand \beta$. Thus by the definition of ($\mand$) we have that for all $\base{E} \baseGeq \baseD$, atomic multisets $Q$ and atomic $p$ if $\alpha\msetsum\beta\suppM{\baseE}{Q}p$ then $\suppM{\base{E}}{K\msetsum N\msetsum Q}p$. In particular, this holds when $\base{E} = \baseD$ and when $Q = M$, so we get $\suppM{\baseD}{K\msetsum N\msetsum M}p$, and thus, we conclude that $\varphi \suppM{\baseC}{K\msetsum M}p$.
            \item $\psi \suppM{\baseC}{K\msetsum M} p$. Follows similarly to the previous case.
        \end{itemize}
        Thus, from the first point, we have that for all $\baseD \baseGeq \baseC$, atomic multisets $V$ and atoms $p$, if $\varphi \suppM{\baseD}{V} p$ and $\psi \suppM{\baseD}{V} p$ then $\suppM{\baseD}{L\msetsum V} p$. Thus, by considering when $\baseD=\baseC$ and $V = K\msetsum M$, we get that $\suppM{\baseC}{L\msetsum K\msetsum M} p$, as required.
    \end{itemize}
    All other cases follow similarly, thus concluding the lemma.
\end{proof}

\noindent Finally, we conclude this Appendix with the remaining cases in the proof of Lemma~\ref{lem:completeness-atomic-definitions}. The statement of the Lemma below contains only the missing cases of Lemma~\ref{lem:completeness-atomic-definitions}.
\begin{lemma*}
    The following hold for all $\baseB\baseGeq\baseILL$ and atomic multisets $L$:
    \begin{enumerate}
        \item $L\deriveBaseM{\baseB}\flatILL{(\varphi\mand\psi)}$ iff for all $\baseC\baseGeq\baseB$, atomic multisets $K$ and atoms $p$, if $\flatILL{\varphi}\msetsum\flatILL{\psi}\msetsum K\deriveBaseM{\baseC}p$ then $L\msetsum K\deriveBaseM{\baseC}p$
        \item $L\deriveBaseM{\baseB}\flatILL{(\varphi\mto\psi)}$ iff $L\msetsum\flatILL{\varphi}\deriveBaseM{\baseB}\psi$
        \item $L\deriveBaseM{\baseB}\flatILL{\mtop}$ iff for all $\baseC\baseGeq\baseB$, atomic multisets $K$ and atoms $p$, if $K\deriveBaseM{\baseC}p$ then $L\msetsum K\deriveBaseM{\baseC}p$
        \item $L\deriveBaseM{\baseB}\flatILL{(\varphi\aand\psi)}$ iff $L\deriveBaseM{\baseB}\flatILL{\varphi}$ and $L\deriveBaseM{\baseB}\flatILL{\psi}$
        \item $L\deriveBaseM{\baseB}\flatILL{(\varphi\aor\psi)}$ iff for all $\baseC\baseGeq\baseB$, atomic multisets $K$ and atoms $p$, if $K\msetsum\flatILL{\varphi}\deriveBaseM{\baseC}p$ and $K\msetsum\flatILL{\psi}\deriveBaseM{\baseC}p$ then $L\msetsum K\deriveBaseM{\baseC}p$
        \item $L\deriveBaseM{\baseB}\flatILL{\top}$ always
        \item $L\deriveBaseM{\baseB}\flatILL{\abot}$ iff $L\msetsum K\deriveBaseM{\baseB}p$, for all atomic multisets $K$ and atoms $p$
    \end{enumerate}
\end{lemma*}
\begin{proof}[(Proof of Lemma~\ref{lem:completeness-atomic-definitions})]
    We take each case in turn:
    \begin{enumerate}
        \item Going left to right, we start by fixing an arbitrary $\baseC$, atomic multiset $K$ and atom $p$ such that  $\flatILL{\varphi}\msetsum\flatILL{\psi}\msetsum K\deriveBaseM{\baseC}p$. We want to show that $L\msetsum K\deriveBaseM{\baseC}p$. We note by monotonicity (Lemma~\ref{lem:monotone-derivability}) that $L\deriveBaseM{\baseC}\flatILL{(\varphi\mand\psi)}$. Thus, we use~\eqref{eq:derive-app} with the rule $\flatILL{\ern\mand}$ to conclude $L\msetsum K\deriveBaseM{\baseC}p$.\\
        Going right to left, we consider the case when $\baseC=\baseB$, $K=\emptymultiset$ and $p=\flatILL{(\varphi\mand\psi)}$. Thus, our hypothesis becomes, if $\flatILL{\varphi}\msetsum\flatILL{\psi}\deriveBaseM{\baseB}\flatILL{(\varphi\mand\psi)}$, then $L\deriveBaseM{\baseB}\flatILL{(\varphi\mand\psi)}$. Since $\flatILL{\varphi}\deriveBaseM{\baseB}\flatILL{\varphi}$ and $\flatILL{\psi}\deriveBaseM{\baseB}\flatILL{\psi}$, both by~\eqref{eq:derive-ref}, then we can use~\eqref{eq:derive-app} with the rule $\flatILL{\irn\mand}$, to conclude that $\flatILL{\varphi}\msetsum\flatILL{\psi}\deriveBaseM{\baseB}\flatILL{(\varphi\mand\psi)}$. Thus, by our hypothesis, we obtain $L\deriveBaseM{\baseB}\flatILL{(\varphi\mand\psi)}$.
        \item Going left to right, we start by supposing that $L\deriveBaseM{\baseB}\flatILL{\varphi\mto\psi}$ and noting that $\flatILL{\varphi}\deriveBaseM{\baseB}\flatILL{\varphi}$ by~\eqref{eq:derive-ref}. Thus, if we use~\eqref{eq:derive-app} with the $\flatILL{\ern\mto}$ rule, we conclude that $L\msetsum\flatILL{\varphi}\deriveBaseM{\baseB}\flatILL{\psi}$.\\
        Going right to left, we immediately use~\eqref{eq:derive-app} with the $\flatILL{\irn\mto}$ rule to obtain $L\deriveBaseM{\baseB}\flatILL{\varphi\mto\psi}$.
        \item Going left to right, we start by fixing an arbitrary $\baseC\baseGeq\baseB$, atomic multiset $K$ and atom $p$ such that $K\msetsum\flatILL{\varphi}\deriveBaseM{\baseC}p$. We wish to show $L\msetsum K\deriveBaseM{\baseC}p$. By monotonicity (Lemma~\ref{lem:monotone-derivability}), we have that $L\deriveBaseM{\baseC}\flatILL{\mtop}$. Thus, we use~\eqref{eq:derive-app} with the $\flatILL{\ern\mtop}$ rule to conclude that $L\msetsum K\deriveBaseM{\baseC}p$.\\
        Going right to left, we consider the case when $\baseC=\baseB$, $K=\emptymultiset$ and $p=\flatILL{\mtop}$. Thus, our hypothesis becomes, if $\deriveBaseM{\baseB}\flatILL{\mtop}$, then $L\deriveBaseM{\baseB}\flatILL{\mtop}$. Since $\deriveBaseM{\baseB}\flatILL{\mtop}$ holds by~\eqref{eq:derive-app}, with the $\flatILL{\irn\mtop}$ rule, we thus conclude $L\deriveBaseM{\baseB}\flatILL{\mtop}$.
        \item Going left to right, we immediately use the $\flatILL{\ern\aand}$ rules on the hypothesis $L\deriveBaseM{\baseB}\flatILL{(\varphi\aand\psi)}$ to conclude that $L\deriveBaseM{\baseB}\flatILL{\varphi}$ and $L\deriveBaseM{\baseB}\flatILL{\psi}$.\\
        Going right to left, since we have that $L\deriveBaseM{\baseB}\flatILL{\varphi}$ and $L\deriveBaseM{\baseB}\flatILL{\psi}$, we use~\eqref{eq:derive-app} with the rule $\flatILL{\irn\aand}$ to obtain $L\deriveBaseM{\baseB}\flatILL{(\varphi\aand\psi)}$.
        \item Going left to right, we start by fixing an arbitrary $\baseC\baseGeq\baseB$, atomic multiset $K$ and atom $p$ such that $K\msetsum\flatILL{\varphi}\deriveBaseM{\baseC}p$ and $K\msetsum\flatILL{\psi}\deriveBaseM{\baseC}p$. Our goal is to show $L\msetsum K\deriveBaseM{\baseC}p$. To this end, we note that, by monotonicity (Lemma~\ref{lem:monotone-derivability}), we have that $L\deriveBaseM{\baseC}\flatILL{(\varphi\aor\psi)}$. Thus, we use~\eqref{eq:derive-app} with the $\flatILL{\ern\aor}$ rule with these hypotheses to obtain $L\msetsum K\deriveBaseM{\baseC}p$.\\
        Going right to left, we consider the case when $\baseC=\baseB$, $K=\emptymultiset$ and $p=\flatILL{(\varphi\aor\psi)}$. Thus, our hypothesis becomes, if $\flatILL{\varphi}\deriveBaseM{\baseB}\flatILL{(\varphi\aor\psi)}$ and $\flatILL{\psi}\deriveBaseM{\baseB}\flatILL{(\varphi\aor\psi)}$, then $L\deriveBaseM{\baseB}\flatILL{(\varphi\aor\psi)}$. Since by~\eqref{eq:derive-ref}, we have that both $\flatILL{\varphi}\deriveBaseM{\baseB}\flatILL{\varphi}$ and $\flatILL{\psi}\deriveBaseM{\baseB}\flatILL{\psi}$, we can thus use~\eqref{eq:derive-app} with both $\flatILL{\irn\aor}$ rules to conclude that indeed $\flatILL{\varphi}\deriveBaseM{\baseB}\flatILL{(\varphi\aor\psi)}$ and $\flatILL{\psi}\deriveBaseM{\baseB}\flatILL{(\varphi\aor\psi)}$. Thus, we conclude that $L\deriveBaseM{\baseB}\flatILL{(\varphi\aor\psi)}$.
        \item By~\eqref{eq:derive-app} with the rule $\flatILL{\irn\top}$, it holds vacuously (due to the presence of the empty additive box) for any $L$, that  $L\deriveBaseM{\baseB}\flatILL{\top}$.
        \item Going left to right, we start by fixing an arbitrary atomic multiset $K$ and atom $p$. We proceed by noting that by~\eqref{eq:derive-app}, with the rule $\flatILL{\ern\abot}$, since $L\deriveBaseM{\baseB}\flatILL{\abot}$, then for any atom $q$ and multiset of atoms $M$, it holds vacuously (due to the presence of the empty additive box) that $L\msetsum M\deriveBaseM{\baseB}q$. Thus, in particular, it holds for $M=K$ and $q=p$, as required.\\
        Going right to left, we are immediately done as we simply consider the case when $K=\emptymultiset$ and $p=\flatILL{\abot}$.
    \end{enumerate}
\end{proof}
\end{appendices}
\section*{Acknowledgements}
I would like to thank Timo Eckhardt, Alex Gheorghiu, Tao Gu, Victor Nascimento, Elaine Pimentel and David Pym for our many discussions on base-extension semantics. In particular, I would also like to thank Katya Piotrovskaya for finding a bug in a previous version of this manuscript. I would like to further thank the anonymous reviewers for their extensive and incredibly helpful comments on earlier drafts of this manuscript. 
\bibliography{PtSILLR2}
\end{document}